\documentclass[a4paper,11pt]{article}
\usepackage[margin=1in]{geometry}

\usepackage{algorithm}
\usepackage[noend]{algorithmic}
\usepackage{hyperref}
\usepackage{pstricks}
\usepackage{enumitem}

\usepackage{times}              
\usepackage{latexsym}
\usepackage{amsmath}
\usepackage{amssymb}
\usepackage{amsfonts}
\usepackage{amsthm}
\usepackage{epsfig}
\usepackage{graphicx}
\usepackage{cite}
\usepackage{caption}

\usepackage{enumerate}

\newtheorem{theorem}{Theorem}[section]
\newtheorem{proposition}[theorem]{Proposition}

\newtheorem{corollary}[theorem]{Corollary}
\newtheorem{definition}[theorem]{Definition}

\DeclareMathOperator*{\argmin}{arg\,min}
\DeclareMathOperator*{\E}{\mathrm{E}}

\newcommand{\A}{\mathcal{A}}
\newcommand{\F}{\mathcal{F}}
\newcommand{\C}{\mathcal{C}}
\newcommand{\D}{\mathcal{D}}

\renewcommand{\S}{\mathcal{S}}
\newcommand{\I}{\mathcal{I}}
\newcommand{\T}{\mathcal{T}}

\newcommand{\U}{\mathcal{U}}
\newcommand{\R}{\mathcal{R}}
\newcommand{\M}{\mathcal{M}}
\renewcommand{\O}{\mathcal{O}}
\renewcommand{\L}{\mathcal{L}}
\renewcommand{\P}{\mathcal{P}}
\newcommand{\e}{\mathcal{E}}

\newcommand{\RBall}{\mathrm{RBall}}

\begin{document}
\title{A Lottery Model for Center-type Problems With Outliers \thanks{Research supported in part by NSF Awards CNS-1010789 and CCF-1422569, and by a research award from Adobe, Inc.  A preliminary version of this work appeared in the \emph{Proc. International Workshop on Approximation Algorithms for Combinatorial Optimization Problems} (APPROX), 2017.}}

\author{  
David G. Harris\thanks{Department of Computer Science, University of Maryland,
College Park, MD 20742. Email:
\texttt{davidgharris29@gmail.com}}
\and
Thomas Pensyl\thanks{Department of Computer Science, University of Maryland, College Park, MD 20742. Email: \texttt{tpensyl@cs.umd.edu}}
\and
Aravind Srinivasan\thanks{Department of Computer Science and Institute for Advanced Computer Studies, University of Maryland, College Park, MD 20742.  Email: 
\texttt{srin@cs.umd.edu}} 
\and Khoa Trinh\thanks{Department of Computer Science, University of Maryland, College Park, MD 20742. Email: \texttt{khoa@cs.umd.edu}}}

\date{}
\maketitle
\begin{abstract}
In this paper, we give tight approximation algorithms for the $k$-center and matroid center problems with outliers. Unfairness arises naturally in this setting: certain clients could always be considered as outliers. To address this issue, we introduce a lottery model in which each client $j$ is allowed to submit a parameter $p_j \in [0,1]$ and we look for a random solution that covers every client $j$ with probability at least $p_j$. Our techniques include a randomized rounding procedure to round a point inside a matroid intersection polytope to a basis plus at most one extra item such that all marginal probabilities are preserved and such that a certain linear function of the variables does not decrease in the process with probability one. 
\end{abstract}

\section{Introduction}

The classic $k$-center and Knapsack Center problems are known to be approximable to within factors of $2$ and $3$ respectively \cite{hoch_shmoys_kc}. These results are best possible unless P=NP \cite{Hsu, hoch_shmoys_kc}. In these problems, we are given a metric graph $G$ and want to find a subset $\S$ of vertices of $G$ subject to either a cardinality constraint or a knapsack constraint such that the maximum distance from any vertex to the nearest vertex in $\S$ is as small as possible. We shall refer to vertices in $G$ as \emph{clients}. Vertices in $\S$ are also called \emph{centers}.

It is not difficult to see that a few \emph{outliers} (i.e., very distant clients) may result in a very large optimal radius in the center-type problems. This issue was raised by Charikar et. al. \cite{charikar_robustkc}, who proposed a \emph{robust} model in which we are given a parameter $t$ and only need to serve $t$ out of given $n$ clients (i.e. $n-t$ \emph{outliers} may be ignored in the solution). Here we consider three robust center-type problems: the Robust $k$-Center (\textsf{RkCenter}) problem, the Robust Knapsack Center (\textsf{RKnapCenter}) problem, and the Robust Matroid Center (\textsf{RMatCenter}) problem.

Formally, an instance $\I$ of the \textsf{RkCenter} problem consists of a set $V$ of vertices, a metric distance $d$ on $V$, an integer $k$, and an integer $t$. Let $n = |V|$ denote the number of vertices (clients). The goal is to choose a set $\S \subseteq V$ of centers (facilities) such that (i) $|\S| \leq k$, (ii) there is a set of \emph{covered} vertices (clients) $\C \subseteq V$ of size at least $t$, and (iii) the objective function
$$ R := \max_{j \in \C} \min_{i \in \S} d(i,j)$$
is minimized. 

In the \textsf{RKnapCenter} problem, each vertex $i \in V$ has a weight $w_i \in [0,1]$, and the cardinality constraint (i) is replaced by the knapsack constraint: $\sum_{i \in \S}w_i \leq 1$. Similarly, in the \textsf{RMatCenter} problem, the constraint (i) is replaced by a matroid constraint: $\S$ must be an independent set of a given matroid $\M$. Here we assume that we have access to the rank oracle of $\M$.

In \cite{charikar_robustkc}, Charikar et. al. introduced a greedy algorithm for the \textsf{RkCenter} problem that achieves an approximation ratio of $3$. Recently, Chakrabarty et. al. \cite{chakrabarty_et_al} give a $2$-approximation algorithm for this problem. Since the $k$-center problem is a special case of the \textsf{RkCenter} problem, this ratio is best possible unless P=NP.

The \textsf{RKnapCenter} problem was first studied by Chen et. al. \cite{jianli_kc}, which showed that one can achieve an approximation ratio of $3$ if allowed to slightly violate the knapsack constraint by a factor of $(1+\epsilon)$. It is still unknown whether there exists a true approximation algorithm for this problem. The current inapproximability bound is still $3$ due to the hardness of the Knapsack Center problem. The current best approximation guarantee for the \textsf{RMatCenter} problem is $7$ by Chen et. al. \cite{jianli_kc}. This problem has a hardness of $(3-\epsilon)$ via a reduction from the $k$-supplier problem. 

From a practical viewpoint, unfairness arises inevitably in the robust model: some clients might always be considered as outliers and hence not covered within the guaranteed radius. To address this issue, we introduce a \emph{lottery model} for these problems. The idea is to randomly pick a solution from a \emph{public list} such that each client $j \in V$ is guaranteed to be covered with probability at least $p_j$, where $p_j \in [0,1]$ is the success rate requested by $j$. This is also motivated by the fact that different clients may have different tolerances of getting connected to their closest facility. One possible way of determining the $p_j$ values is to let each client $j$ pay for the chance of being served. Note that the robust model is a special case when $p_j = 0$ for all $j \in V$. Similarly, when all $p_j$'s are equal to $1$, it becomes the standard model where all clients must be connected.

The lottery model might also be useful in the context of clustering. Recall that clustering is a fundamental task in unsupervised machine learning. Basically, we want to partition a set of data points to clusters in such a way that the points in the same cluster are ``similar'' to each other. The $k$-center clustering is one of the popular approaches to this task. (See also $k$-means clustering \cite{svensson, kanungo} and $k$-median clustering \cite{byrka_kmed, li_svensson, book:ws}.) Naturally, the robust model can be applied to get rid of some ``bad data points'' or ``noise'' and hence improve the overall quality of all clusters. Here over-fitting may occur when this model misclassifies some good points as outliers. The lottery model offers the flexibility to decide whether a point should be included in the solution via the $p_j$ values.

In this paper, we introduce new approximation algorithms for these robust center problems under the lottery model. (Note that this model has been used recently for the Knapsack Center problems  \cite{srdr}, although the techniques and problems in that paper are different from ours.) We also propose improved approximation algorithms for the \textsf{RkCenter} problem and the \textsf{RMatCenter} problem.

\subsection{The Lottery Model}
In this section, we formally define our lottery model for the above-mentioned problems. First, the \emph{Fair} Robust $k$-Center (\textsf{FRkCenter}) problem is formulated as follows. Besides the parameters $V,d,k$ and $t$, each vertex $j \in V$ has a ``target'' probability $p_j \in [0,1]$. We are interested in the minimum radius $R$ for which there exists a distribution $\D$ on subsets of $V$ such that a set $\S$ drawn from $\D$ satisfies the following constraints:
\begin{description}
	\item[Coverage constraint:] $|\C| \geq t$ with probability one, where $\C$ is the set of all clients in $V$ that are within radius $R$ from some center $\S$,
	\item[Fairness constraint:] $\Pr[j \in \C] \geq p_j$ for all $j \in V$, where $\mathcal{C}$ is as in the coverage constraint,
	\item[Cardinality constraint:] $|\S| \leq k$ with probability one.
\end{description}
Here we aim for a polynomial-time, randomized algorithm that can sample from $\D$. Note that the \textsf{RkCenter} is a special of this variant in which all $p_j$'s are set to be zero.

The \emph{Fair} Robust Knapsack Center (\textsf{FRKnapCenter}) problem and \emph{Fair} Robust Matroid Center (\textsf{FRMatCenter}) problem are defined similarly except that we replace the cardinality constraint by a knapsack constraint and a matroid constraint, respectively. More formally, in the \textsf{FRKnapCenter} problem, each vertex $i$ has a weight $w_i \in [0,1]$, and we require the total weight of centers in $\S$ to be at most $1$ with probability one. Similarly, in the \textsf{FRMatCenter} problem, we are given a matroid $\M$ and we require the solution $\S$ to be an independent set of $\M$ with probability one.

\subsection{Our contributions and techniques}
First of all, we give tight approximation algorithms for the \textsf{RkCenter} and \textsf{RMatCenter} problems.

\begin{theorem}
There exist a $2$-approximation algorithm for the \textsf{RkCenter} problem \footnote{A 2-approximation algorithm has also been found independently by Chakrabarty et. al. \cite{chakrabarty_et_al}, and in a private discussion between Marek Cygan and Samir Khuller. Our algorithm here is different from the algorithm in \cite{chakrabarty_et_al}.} and a $3$-approximation algorithm for the \textsf{RMatCenter} problem. 
\label{theorem:standard_models}
\end{theorem}

Our main results for the lottery model are summarized in the following theorems.

\begin{theorem}
For any given constant $\epsilon>0$ and any instance $\I = (V, d, k, t, \vec{p})$ of the \textsf{FRkCenter} problem, there is a randomized polynomial-time algorithm $\A$  which can compute a random solution $\S$ such that
\begin{itemize}
	\item $|\S| \leq k$ with probability one,
	\item $|\C| \geq (1-\epsilon)t$, where $\C$ is the set of all clients within radius $2R$ from some center in $\S$ and $R$ is the optimal radius,
	\item $\Pr[j \in \C] \geq (1-\epsilon) p_j$ for all $j \in V$.
\end{itemize}

\label{thm:FRkCenter}
\end{theorem}

\begin{theorem}
For any $\epsilon > 0$ and any instance $\I = (V, d, w, B, t, \vec{p})$ of the \textsf{FRKnapCenter} problem, there is a randomized polynomial-time algorithm $\A$ which can return random solution $\S$ such that 
\begin{itemize}
	\item $\sum_{i \in \S} w_i \leq (1+\epsilon)B$ with probability one,
	\item $|\C| \geq t$, where $\C$ is the set of vertices within distance $3R$ of $\S$.
	\item $\Pr[j \in \C] \geq p_j$ for all $j \in V$.
\end{itemize}  
\label{thm:FRKnapCenter}
\end{theorem}

Finally, the \textsf{FRMatCenter} can be reduced to (randomly) rounding a point in a matroid intersection polytope. We design a randomized rounding algorithm which can output a pseudo solution, which consists of a basis plus one extra center. By using a preprocessing step and  a configuration LP, we can satisfy the matroid constraint exactly (respectively, knapsack constraint) while slightly violating the coverage and fairness constraints in the \textsf{FRMatCenter} (respectively, \textsf{FRKnapCenter}) problem. We believe these techniques could be useful in other facility-location problems  (e.g., the matroid median problem \cite{krishnaswamy_km, swamy_km}) as well. 

\begin{theorem}
For any given constant $\gamma > 0$ and any instance $\I = (V, d, \M, t, \vec{p})$ of the \textsf{FRMatCenter} (respectively, \textsf{FRKnapCenter}) problem, there is a randomized polynomial-time algorithm $\A$ which can return a random solution $\S$ such that 
\begin{itemize}
	\item $\S$ is a basis of $\M$ with probability one, (respectively, $w(\S) \leq B$ with probability one)
	\item $|\C| \geq t - \gamma^2 n$, where $\C$ is the set of vertices within distance $3R$ of $\S$.
	\item there exists a set $T \subseteq V$ of size at least $(1 - \gamma)n$, which is deterministic, such that $\Pr[j \in \C] \geq p_j - \gamma$ for all $j \in T$.
\end{itemize} 
\label{thm:FRMatCenter}
\end{theorem}

\subsection{Organization}
The rest of this paper is organized as follows. In Section 2, we review some basic properties of matroids and discuss a filtering algorithm which is used in later algorithms. Then we develop approximation algorithms for the \textsf{FRkCenter}, \textsf{FRKnapCenter}, and \textsf{FRMatCenter} problems in the next three sections.

\section{Preliminaries}
\subsection{Matroid polytopes}
We first review a few basic facts about matroid polytopes. For any vector $z$ and set $S$, we let $z(S)$ denote the sum $\sum_{i \in S}z_i$. Let $\M$ be any matroid on the ground set $\Omega$ and $r_\M$ be its rank function. The matroid base polytope of $\M$ is defined by 
$$ \P_\M := \left\{x \in \mathbb{R}^{\Omega}: x(S) \leq r_\M(S) ~~ \forall S \subseteq \Omega; ~~~x(\Omega) = r_\M(\Omega);~~~ x_i \geq 0 ~~\forall i \in \Omega \right\}. $$

\begin{definition} Suppose $A x \leq b$ is a valid inequality of $\P_\M$. A face $D$ of $\P_\M$ (corresponding to this valid inequality) is the set $ D := \left\{x \in \P_\M: A x = b \right\}.	$
\end{definition}

The following theorem gives a characterization for any face of $\P_\M$ (See, e.g., \cite{schrijver, Lau2001}). 
\begin{theorem}
Let $D$ be any face of $\P_\M$. Then it can be characterized by
$$ D = \left\{x \in \mathbb{R}^{\Omega}: x(S) = r_\M(S)~~\forall S \in \L; ~~~ x_i = 0 ~~\forall i \in J; ~~~ x \in \P_\M  \right\},$$
where $J \subseteq \Omega$ and $\L$ is a chain family of sets:
$L_1 \subset L_2 \subset \ldots \subset L_m.$ Moreover, it is sufficient to choose $\L$ as any maximal chain $L_1 \subset L_2 \subset \ldots \subset L_m$ such that $x(L_i) = r_\M(L_i)$ for all $i = 1, 2, \ldots, m$.
\label{theorem:decompose_face}
\end{theorem}

\begin{proposition} Let $x \in \P_\M$ be any point and $I$ be the set of all tight constraints of $\P_\M$ on $x$. Suppose $D$ is the face with respect to $I$. Then one can compute a chain family $\L$ for $D$ as in Theorem \ref{theorem:decompose_face} in polynomial time.
\label{prop:decompose_face}
\end{proposition}
\begin{proof} Recall that $r_\M$ is a submodular function. Thus the function $r'_\M(S) = r_\M(S) - x(S)$ for $S \subseteq \Omega$ is also submodular. It is well-known that submodular minimization can be done in polynomial time. We solve the following optimization problem: $\min\left\{r'_\M(S): S \subseteq \Omega \right\}$. If there are multiple solutions, we let $S_0$ be any solution of minimal size. (This can be done easily, say, by trying to drop each item from the current solution and resolving the program.) We add $S_0$ to our chain. Then we find some \emph{minimal} superset $S_1$ of $S_0$ such that $r'_\M(S_1) = 0$, add $S_1$ to our chain, and repeat the process.
\end{proof}

\begin{corollary} Let $D$ be any face of $\P_\M$. Then it can be characterized by
$$ D = \left\{x \in \mathbb{R}^{\Omega}: x(S) = b_S~~\forall S \in \O; ~~~ x_i = 0 ~~\forall i \in J; ~~~ x \in \P_\M  \right\},$$
where $J \subseteq \Omega$ and $\O$ is a family of pairwise disjoint sets:
$O_1, O_2, \ldots, O_m,$ and $b_{O_1}, \ldots, b_{O_m}$ are some integer constants. 
\label{cor:decompose_face}
\end{corollary}
\begin{proof}
By Theorem \ref{theorem:decompose_face}, we have
$$ D = \left\{x \in \mathbb{R}^{\Omega}: x(S) = r_\M(S)~~\forall S \in \L; ~~~ x_i = 0 ~~\forall i \in J; ~~~ x \in \P_\M  \right\},$$
where $J \subseteq \Omega$ and $\L$ is the chain: $L_1 \subset L_2 \subset \ldots \subset L_m$. Now let us define $O_1 := L_1, O_2 := L_2 \setminus L_1, O_3 := L_3 \setminus L_2, \ldots, O_m := L_m \setminus L_{m-1}$, and $b_{O_1} := r_\M(L_1), b_{O_2} := r_\M(L_2) - r_\M(L_1), \ldots,  b_{O_m} := r_\M(L_m) - r_\M(L_{m-1})$. It is not difficult to verify that
$$ D = \left\{x \in \mathbb{R}^{\Omega}: x(S) = b_S~~\forall S \in \O; ~~~ x_i = 0 ~~\forall i \in J; ~~~ x \in \P_\M  \right\}.$$
\end{proof}

\subsection{Filtering algorithm}
All algorithms in this paper are based on rounding an LP solution. In general, for each vertex $i \in V$, we have a variable $y_i \in [0,1]$ which represents the probability that we want to pick $i$ in our solution. (In the standard model, $y_i$ is the ``extent'' that $i$ is opened.) In addition, for each pair of $i,j \in V$, we have a variable $x_{ij} \in [0,1]$ which represents the probability that $j$ is connected to $i$. 

Note that in all center-type problems, the optimal radius $R$ is always the distance between two vertices. Therefore, we can always ``guess'' the value of $R$ in $O(\log n)$ by a binary search on the sorted entries of the matrix $d$. WLOG, we may assume that we know the correct value of $R$.  For any $j \in V$, we define $B_j := \{ i \in V: d(i,j) \leq R \}$ and we let $F_j := \{i \in V: d(i,j) \leq R \wedge x_{ij} > 0\}$ and $s_j := \sum_{i \in B_j} x_{ij}$. We shall refer to $F_j$ as a cluster with cluster center $j$. Depending on a specific problem, we may have different constraints on $x_{ij}$'s and $y_i$'s. In general, the following constraints are valid in most of the problems here:
{\allowdisplaybreaks
\begin{align}
	\sum_{j \in V} \sum_{i \in B_j} x_{ij} &\geq t,  \label{ineq:cov} \\
	\sum_{i \in B_j} x_{ij} &\leq 1, \quad \forall j \in V, \label{ineq:connect} \\
	x_{ij} &\leq y_i,  \label{ineq:open} \quad \forall i,j \in V, \\
	y_i, x_{ij} &\geq 0,  \quad \forall i,j \in V. \label{ineq:nonneg}
\end{align}
}
For the \emph{fair} variants, we may also require that
\begin{align}
	\sum_{i \in B_j} x_{ij} &\geq p_j, \quad \forall j \in V. \label{ineq:fair} 
\end{align}

Constraint (\ref{ineq:cov}) says that at least $t$ vertices should be covered. Constraint (\ref{ineq:connect}) ensures that each vertex is only connected to at most one center. Constraint (\ref{ineq:open}) means vertex $j$ can only connect to center $i$ if it is open. Constraint (\ref{ineq:fair}) says that the total probability of $j$ being connected should be at least $p_j$. By constraints (\ref{ineq:connect}) and (\ref{ineq:open}), we have $y(F_j) \leq 1$.

The first step of all algorithms in this paper is to use the following \emph{filtering} algorithm to obtain a maximal collection of disjoint clusters. The algorithm will return the set $V'$ of cluster centers of the chosen clusters. In the process, we also keep track of the number $c_j$ of other clusters removed by $F_j$ for each $j \in V'$.

\begin{algorithm}[h]
\caption{$\textsc{RFiltering}\left(x,y \right)$}
\begin{algorithmic}[1]
\STATE $V' \gets \emptyset$
\FOR{\textbf{each} unmarked cluster $F_j$ in \textbf{decreasing order} of $s_j$}
	\STATE $V' \gets V' \cup \{j\}$
	\STATE Set all unmarked clusters $F_k$ (including $F_j$ itself) s.t. $F_k \cap F_j \neq \emptyset$ as marked.
	\STATE Let $c_j$ be the number of marked clusters in this step. 
\ENDFOR
\STATE $\vec{c} \gets (c_j: j \in V')$
\RETURN $(V', \vec{c})$
\end{algorithmic} 
\end{algorithm}

\section{The $k$-center problems with outliers}
In this section, we first give a simple $2$-approximation algorithm for the \textsf{RkCenter} problem. Then, we give an approximation algorithm for the \textsf{FRkCenter} problem, proving Theorem \ref{thm:FRkCenter}.

\subsection{The robust $k$-center problem}
Suppose $\I = (V,d,k,t)$ is an instance the \textsf{RkCenter} problem with the optimal radius $R$. Consider the polytope $\P_\textsf{RkCenter}$ containing points $(x, y)$ satisfying constraints (\ref{ineq:cov})--(\ref{ineq:nonneg}), and the cardinality constraint: 
\begin{align}
	\sum_{i \in V} y_i &\leq k. \label{ineq:cardinality} 
\end{align}
Since $R$ is the optimal radius, it is not difficult to check that $\P_\textsf{RkCenter} \neq \emptyset$. Let us pick any fractional solution $(x,y) \in \P_\textsf{RkCenter}$. The next step is to round $(x,y)$ into an integral solution using the following simple algorithm.

\begin{algorithm}[h]
\caption{$\textsc{RkCenterRound}\left(x,y \right)$}
\begin{algorithmic}[1]
\STATE $(V',\vec{c}) \gets \textsc{RFiltering}\left(x,y \right).$
\STATE $\S \gets $ the top $k$ vertices $i \in V'$ with highest value of $c_i$.
\RETURN $\S$
\end{algorithmic} 
\end{algorithm}

\textbf{Analysis.} By construction, the algorithm returns a set $\S$ of $k$ open centers. Note that, for each $i \in \S$, $c_i$ is the number of distinct clients within radius $2R$ from $i$. Thus, it suffices to show that $\sum_{i \in \S} c_i \geq t$. 
By inequality (\ref{ineq:connect}), we have $s_j \leq 1$ for all $j \in V'$. Thus,
$$  \sum_{i \in V'} c_i s_i \geq \sum_{i \in V} s_i \geq t, $$
where the first inequality is due to the greedy choice of vertices in $V'$ and the second inequality follows from (\ref{ineq:cov}). Now recall that the clusters whose centers in $V'$ are pairwise disjoint. By constraint (\ref{ineq:cardinality}),
$$\sum_{i \in V'} s_i \leq \sum_{i \in V'} y(F_i) \leq \sum_{i \in V} y_i \leq k.$$ 
It follows by the choice of $\S$ that $\sum_{i \in \S} c_i \geq t$. This concludes the first part of Theorem \ref{theorem:standard_models}.

\subsection{The fair robust $k$-center problem}

Assume $\I = (V,d,k,t, \vec{p})$ be an instance of the \textsf{FRkCenter} problem with the optimal radius $R$. Fix any $\epsilon > 0$. If $k \leq 2/\epsilon$, then we can generate all possible $O\left(n^{1/\epsilon}\right)$ solutions and then solve an LP to obtain the corresponding marginal probabilities. So the problem can be solved easily in this case. We will assume that $k \geq 2/\epsilon$ for the rest of this section. Consider the polytope $\P_\textsf{FRkCenter}$  containing points $(x, y)$ satisfying constraints (\ref{ineq:cov})--(\ref{ineq:nonneg}), the fairness constraint (\ref{ineq:fair}), and the cardinality constraint (\ref{ineq:cardinality}). We now show that $\P_\textsf{FRkCenter}$ is actually a valid relaxation polytope.
\begin{proposition} The polytope $\P_\textsf{FRkCenter}$ is non-empty.
\label{prop:frkcenter}
\end{proposition}
\begin{proof}
It suffices to point out a solution $(x,y) \in \P_\textsf{FRkCenter}$. Since $R$ is the optimal radius, there exists a distribution $\D$ satisfying the coverage, fairness, and cardinality constraints. Suppose $\S$ is sampled from $\D$ and $\C$ is the set of all clients in $V$ that are within radius $R$ from some center $\S$. We now set $y_i := \Pr[i \in \S]$ for all $i \in V$. Since $|\S| \leq k$ with probability one, we have $\sum_{i \in V} y_i = \E[|\S|] \leq k$, and hence constraint (\ref{ineq:cardinality}) is valid. 

We construct the assignment variable $x$ as follows. For each $j \in V$, set $z_j := 0$. Then for each $i \in B_j$, set $x_{ij} := \min\{y_i, 1-z_j\}$ and update $z_j := z_j + x_{ij}$. It is not hard to see that inequalities (\ref{ineq:connect}) and (\ref{ineq:open}) hold by this construction. Now let us fix any $j \in V$. By fairness guarantee of $\D$ and the union bound, we have
$$p_j \leq \Pr[j \in \C] \leq \sum_{i \in B_j} y_i.$$
Thus, by construction of $x$, we have
$$\sum_{i \in B_j} x_{ij} \geq \Pr[j \in \C] \geq p_j,$$
and hence inequality (\ref{ineq:fair}) is satisfied. Finally, we have
$$ \E[|\C|] = \sum_{j \in V} \Pr[j \in \C] \leq \sum_{j \in V} \sum_{i \in B_j} x_{ij}.$$
Since $|\C| \geq t$ with probability one, $\E[|\C|] \geq t$, implying that inequality (\ref{ineq:cov}) holds.

\end{proof}

Our algorithm is as follows.

\begin{algorithm}[h]
\caption{$\textsc{FRkCenterRound}\left(\epsilon,x,y \right)$}
\begin{algorithmic}[1]
\STATE $(V',\vec{c}) \gets \textsc{RFiltering}\left(x,y \right).$

\STATE{Form vector $y'$ by setting}
$$
y_j' \gets (1 - \epsilon)\sum_{i \in F_j} x_{ij}
$$
for each $j \in V'$.

\WHILE{at least three entries of $y'$ are in the range $(0,1)$}
	\STATE Let $\delta \in \mathbb{R}^{V'}, \delta \neq 0$ be such that $\delta_i = 0 ~~\forall i \in V':y'_i \in \{0,1\}$,  $\delta(V') = 0$, and $\vec{c} \cdot \delta = 0.$
	\STATE Choose scaling factors $a,b > 0$ such that 
		\begin{itemize}[noitemsep,nolistsep] 
			\item $y'+a\delta\in [0,1]^{V'}$ and $y'-b\delta \in [0,1]^{V'}$
			\item there is at least one new entry of $y'+a\delta$ which is equal to zero or one
			\item there is at least one new entry of $y'-b\delta$ which is equal to zero or one
		\end{itemize}
	\STATE With probability $\frac{b}{a+b}$, update $y' \gets y' + a\delta$; else, update $y' \gets y' - b\delta$.	
\ENDWHILE

\RETURN $\S = \{i \in V: y_i' > 0\}$.
\end{algorithmic} 
\label{algo:frkcenter}
\end{algorithm}

\bigskip \textbf{Analysis.} First, note that one can find such a vector $\delta$ in line $4$ as the system of $\delta(V')=0$ and $\vec{c} \cdot \delta = 0$ consists of two constraints and at least three variables (and hence is underdetermined.) Thus, the algorithm rounds at least one fractional variable per iteration, and terminates after $O(n)$ rounds. Let $\S$ denote the (random) solution returned by  $\textsc{FRkCenterRound}$ and $\C$ be the set of all clients within radius $3R$ from some center in $\S$. Theorem \ref{thm:FRkCenter} can be verified by the following propositions.

\begin{proposition}  $|\S| \leq k$ with probability one.
\end{proposition}
\begin{proof} By definition of $y'$ at line $2$ of \textsc{FRkCenterRound}, we have
\begin{align*}
y'(V') &= (1-\epsilon) \sum_{j \in V'} \sum_{i \in F_j} x_{ij}  \leq (1-\epsilon) \sum_{j \in V'} \sum_{i \in F_j} y_i \leq (1-\epsilon)k \leq k-2,	
\end{align*}
since $k \geq 2/\epsilon$.
Note that the sum $y'(V')$ is never changed in the while loop (lines $4 \ldots 7$) because $\delta(V') = 0$. Then the final vector $y'$ contains at most two fractional values at the end of the while loop. By rounding these two values to one, the size of $\S$ is indeed at most $k$.
\end{proof}

\begin{proposition} $|\C| \geq (1-\epsilon)t$ with probability one.
\end{proposition}
\begin{proof}
At the beginning of the while loop, we have
\begin{align*}
	\vec{c} \cdot y' = \sum_{j \in V'} c_j y'_j(F_j) &= (1-\epsilon) \sum_{j \in V'} c_j s_j \geq (1-\epsilon) \sum_{j \in V}  s_j \geq (1-\epsilon)t.
\end{align*}
Again, the quantity $\vec{c} \cdot y'$ is unchanged in the while loop because $\vec{c} \cdot \delta = 0$ implies that $\vec{c} \cdot (y'+a\delta) = \vec{c} \cdot y'$ and $\vec{c} \cdot (y'-b\delta) = \vec{c} \cdot y'$ in each iteration. Note that if $y' \in \{0,1\}^{V'}$, then $\vec{c} \cdot y'$ is the number of clients within radius $2R$ from some center $i$ such that $y'_i = 1$. Basically, we round the two remaining fractional values of $y'$ to one in line $8$; and hence, the dot product should be still at least $(1-\epsilon)t$.
\end{proof}

\begin{proposition} $\Pr[j \in \C] \geq (1-\epsilon)p_j$ for all $j \in V$.
\end{proposition}
\begin{proof}
Fix any $j \in V$. The algorithm \textsc{RFiltering} guarantees that there exists $k \in V'$ such that $F_j \cap F_k \neq \emptyset$ and $s_k \geq s_j$. Now we claim that $\E[y'_k] = y'_k$. This is because the expected value of $y'_k$ does not change after any single iteration: 
\begin{align*}
	\E[y'_k] = (y'_k+a\delta)\frac{b}{a+b} + (y'_k-b\delta)\frac{a}{a+b} = y'_k.
\end{align*}
Then
$$\Pr[k \in \S] = \Pr[y'_k > 0] \geq \Pr[y'_k = 1] = \E[y'_k] = y'_k = (1-\epsilon)s_k.$$ 
Therefore,
\begin{align*}
	\Pr[j \in \C] \geq \Pr[k \in \S] \geq (1-\epsilon)s_k \geq (1-\epsilon)s_j \geq (1-\epsilon)p_j,
\end{align*}
by constraint (\ref{ineq:fair}).

\end{proof}

\section{The Knapsack Center problems with outliers}
We study the \textsf{RKnapCenter} and \textsf{FRKnapCenter} problems in this section. Recall that in these problems, each vertex has a weight and we want to make sure that the total weight of the chosen centers does not exceed $1$. We first give a $3$-approximation algorithm for the \textsf{RKnapCenter} problem that slightly violates the knapsack constraint. Although this is not better than the known result by \cite{jianli_kc}, both our algorithm and analysis here are more natural and simpler. It serves as a starting point for the next results. For the \textsf{FRKnapCenter}, we show that it is possible to satisfy the knapsack constraint exactly with small violations in the coverage and fairness constraints.

\subsection{The robust knapsack center problem}
\label{sec:rknapcenter}
Suppose $\I = (V, d, w,  t)$ is an instance the \textsf{RKnapCenter} problem with the optimal radius $R$. Consider the polytope $\P_\textsf{RKnapCenter}$  containing points $(x, y)$ satisfying constraints (\ref{ineq:cov})--(\ref{ineq:nonneg}), and the knapsack constraint: 
\begin{align}
	\sum_{i \in V} w_iy_i \leq 1. \label{ineq:knap} 
\end{align} 

Again, it is not difficult to check that $\P_\textsf{RKnapCenter} \neq \emptyset$. Let us pick any fractional solution $(x,y) \in \P_\textsf{RKnapCenter}$. Our pseudo-approximation algorithm to round $(x,y)$ is as follows.

\begin{algorithm}[h]
\caption{$\textsc{RKnapCenterRound}\left(x,y \right)$}
\begin{algorithmic}[1]
\STATE $(V',\vec{c}) \gets \textsc{RFiltering}\left(x,y \right).$
\STATE For each $i \in V'$, let $v_i = \argmin_{j \in F_i}\{w_j\}$ be the vertex with smallest weight in $F_i$
\STATE Let $\P' := \left\{z \in [0,1]^{V'}: \sum_{i\in V'}c_i z_i \geq t ~~\wedge ~~ \sum_{i \in V'} w_{v_i} z_i \leq 1 \right\}$
\STATE Compute an extreme point $Y$ of $\P'$
\RETURN $\S = \{v_i: i \in V,~ Y_i > 0\}$
\end{algorithmic} 
\label{algo:rknapcenter}
\end{algorithm}

\bigskip \noindent \textbf{Analysis.} We first claim that $\P' \neq \emptyset$ which implies that the extreme point $Y$ of $\P'$ (in line $4$) does exist. To see this, we claim that the vector $s$ lies in $\P'$. For, we have:
$$
\sum_{i\in V'} c_i s_i \geq \sum_{i\in V} s_i \geq t. $$
Also,
\begin{align*}
	\sum_{i \in V'} w_{v_i} s_i &= \sum_{i \in V'} w_{v_i} \sum_{j \in F_i} x_{ji} \leq \sum_{i \in V'} w_{v_i} \sum_{j \in F_i} y_j \leq \sum_{i \in V'}  \sum_{j \in F_i} w_j y_j  \leq \sum_{i \in V} w_i y_i \leq 1. 
\end{align*}
All the inequalities follow from LP constraints and definitions of $s_i, c_i, $ and $v_i$.

\begin{proposition}
The solution $\S$ returned by \textsc{RKnapCenterRound} satisfies $w(\S) \leq 1 + 2w_{\text{max}}$ and $|\C| \geq t$, where $\C$ is the set of vertices within distance $3R$ of $\S$ and $w_{\text{max}}$ is the maximum weight of any vertex in $V$.
\label{prop:RKnapCenterRound}
\end{proposition}
\begin{proof}
First, observe that any extreme point of $\P'$ has at most two fractional values. (In the worst case, an extreme point $z$ is fully determined by $|V'|-2$ tight constraints of the form $z_i = 0$ or $z_i = 1$, $\sum_{i\in V'}c_i z_i = t$, and $ \sum_{i \in V'} w_{v_i} z_i = 1$.) By construction of $\S$, there are at most two vertices $i^*, i^{**}$ such that $Y_{i^*}, Y_{i^{**}}$ are fractional. Thus, 
$$w(\S) = \sum_{i \in \S \setminus\{i^*, i^{**}\}} w_{v_i}Y_i + w_{i^*} + w_{i^{**}} \leq 1 + 2w_{\text{max}}.$$

Recall that each $i \in V'$ has $c_i$ clients within distance $2 R$ (and each client is counted only one time). By the triangle inequality, these clients are within distance $3R$ from $v_i$. Thus, $\S$  covers at least 
$$ \sum_{i \in \S \setminus\{i^*, i^{**}\}} c_i Y_i + c_{i^*} + c_{i^{**}} \geq \sum_{i \in \S} c_i Y_i \geq t,$$
clients within radius $3R$.
\end{proof}

\subsection{The fair robust knapsack center problem}
In this section, we will first consider a simple algorithm that only violates the knapsack constraint by two times the maximum weight of any vertex. Then using a configuration polytope to ``condition'' on the set of ``big'' vertices, we show that it is possible to either violate the budget by $(1+\epsilon)$ or to preserve the knapsack constraint while slightly violating the coverage and fairness constraints.

\subsubsection{Basic algorithm}

Suppose $\I = (V, d, w,  t, \vec{p})$ is an instance the \textsf{FRKnapCenter} problem with the optimal radius $R$. Consider the polytope $\P_\textsf{FRKnapCenter}$  containing points $(x, y)$ satisfying constraints (\ref{ineq:cov})--(\ref{ineq:nonneg}), the fairness constraint (\ref{ineq:fair}), and the knapsack constraint (\ref{ineq:knap}). The proof that $\P_\textsf{FRKnapCenter} \neq \emptyset$ is very similar to that of Proposition \ref{prop:frkcenter} and is omitted here.

The following algorithm is a randomized version of 
\textsc{RKnapCenterRound}.

\begin{algorithm}[h]
\caption{$\textsc{BasicFRKnapCenterRound}\left(x,y \right)$}
\begin{algorithmic}[1]
\STATE $(V',\vec{c}) \gets \textsc{RFiltering}\left(x,y \right).$
\STATE For each $i \in V'$ let $v_i := \argmin_{j \in F_i}\{w_j\}$ be the vertex with smallest weight in $F_i$
\STATE Let $\P' := \left\{z \in [0,1]^{V'}: \sum_{i\in V'}c_i z_i \geq t ~~\wedge ~~ \sum_{i \in V'} w_{v_i} z_i \leq 1 \right\}$
\STATE Decompose the vector $s$ as a convex combination of extreme points $z^{(1)}, \ldots, z^{(n+1)}$ of $\P'$:
$$ s = p_1 z^{(1)} + \ldots + p_{n+1} z^{(n+1)},$$
where $\sum_\ell p_\ell = 1$ and $p_\ell \geq 0$ for all $\ell \in [n+1]$. 
\STATE Randomly choose $Y \gets z^{(\ell)}$ with probability $p_\ell$.
\RETURN $\S = \{v_i: i \in V,~ Y_i > 0\}$
\end{algorithmic} 
\label{algo:basicrfknapcenter}
\end{algorithm}

Note that $\P' \neq \emptyset$, and so the decomposition at line 4 is well-defined (see the analysis in Section \ref{sec:rknapcenter}).

\begin{proposition}
The algorithm \textsc{BasicFRKnapCenterRound} returns a random solution $\S$ such that $w(\S) \leq 1 + 2w_{\text{max}}$, $|\C| \geq t$, and $\Pr[j \in \C] \geq p_j$ for all $j \in V$, where $\C$ is the set of vertices within distance $3R$ of $\S$ and $w_{\text{max}}$ is the maximum weight of any vertex in $V$.
\label{prop:basicrfknapcenter}
\end{proposition}
\begin{proof}
By similar arguments to the proof of Proposition \ref{prop:RKnapCenterRound}, we have $w(\S) \leq 1 + 2w_{\text{max}}$ and $|\C| \geq t$. To obtain the fairness guarantee, observe that $v_i \in \S$ with probability at least $z_i = s_i$. For any $j \in V$, let $k \in V'$ be the vertex that removed $j$ in the filtering step. We have
\begin{align*}
	\Pr[j \in \C] \geq \Pr[v_k \in \S] \geq s_k \geq s_j \geq p_j,
\end{align*}
where the penultimate inequality is due to our greedy choice of $k$ in \textsc{RFiltering}.
\end{proof}

\subsubsection{An algorithm slightly violating the budget constraint}
Fix a small parameter $\epsilon > 0$. A vertex $i$ is said to be \emph{big} iff $w_i > \epsilon$. Let $\U$ denote the collection of all possible sets of big vertices. Since a solution can contain at most $1/\epsilon$ big vertices, we have $|\U| \leq n^{O(1/\epsilon)}$. Consider the \emph{configuration} polytope $\P_\textsf{config1}$ containing points $(x,y,q)$ with the following constraints:
\begin{align*}
  \begin{cases}
      \sum_{U \in \U}q_U = 1     & \quad \\
      \sum_{i \in B_j} x^{U}_{ij} \leq q_U    & \quad \forall j \in V, U \in \U \\
      \sum_{U \in \U} \sum_{i \in B_j} x^{U}_{ij} \geq p_j    & \quad \forall j \in V \\
       x^{U}_{ij} \leq y^{U}_i   & \quad \forall i,j \in V, U \in \U\\
       \sum_{i \in V} w_i y^{U}_i \leq q_U B   & \quad  \forall U \in \U \\
       \sum_{j \in V} \sum_{i \in B_j} x^{U}_{ij} \geq q_U t   & \quad \\
       y^{U}_i = 1 & \quad \forall U \in \U, i \in U\\
       y^{U}_i = 0 & \quad \forall U \in \U, i \in V \setminus U, w_i > 1/\epsilon \\
       x^{U}_{ij}, y^{U}_i, q_U \geq 0   & \quad \forall i,j \in V, U \in \U \\
  \end{cases}
\end{align*}

We first claim that $\P_\textsf{config1}$ is a valid relaxation polytope for the problem.

\begin{proposition} The polytope $\P_\textsf{config1}$ is non-empty.
\label{prop:relaxation_knapcenter}
\end{proposition}
\begin{proof}
Fix any optimal distribution $\D$. Suppose $\S$ is sampled from $\D$. For any $U \in \U$, let $\e(U)$ be the event that $U \subseteq \S$ and $\S \setminus U$ contains no big vertex, and let $q_U = \Pr[ \e(U) ]$. It is clear that $\sum_{U \in \U} q_U = 1$. Let $x^{U}_{ij} = \Pr[ \e(U) \wedge \text{ $j$ is connected to $i$}]$ and let $y^{U}_i = \Pr[ \e(U) \wedge i \in \S]$.

Now observe that
\begin{align*}
	q_U = \Pr[\e(U)] = \sum_{i \in B_j} \Pr[ j \text{ is connected to }i \wedge  \e(U) ] = \sum_{i \in B_j} x^{U}_{ij}.
\end{align*}
Similarly,
\begin{align*}
	p_j &\leq \Pr[j \text{ is connected}] = \sum_{U \in \U} \Pr[j \text{ is connected} \wedge \e(U)] \\
		&= \sum_{U \in \U} \sum_{i \in B_j} \Pr[j \text{ is connected to }i \wedge \e(U)] = \sum_{U \in \U} \sum_{i \in B_j} x^{U}_{ij}.
\end{align*}

Note that $x^{U}_{ij}/q_U$ and $y^{U}_i/q_U$ are the probabilities that $j$ is connected to $i$ and $i \in \S$ conditioned on $\e(U)$, respectively. Since the number of connected clients is at least $t$ with probability one, we have
\begin{align*}
	t &\geq \E[\text{\# connected clients} | \e(U)] = \sum_{j \in V} \sum_{i \in B_j} \Pr[j \text{ is connected to } i | \e(U)] = \sum_{j \in V} \sum_{i \in B_j} x^{U}_{ij} / q_U.
\end{align*}
Similarly, $w(\S) \leq 1$ with probability one and so
\begin{align*}
	1 \geq \E[w(\S) | U] =  \sum_{i \in V} w_i (y^{U}_i/q_U). 
\end{align*}
The other constraints can be verified easily. We conclude that $(x,y,q) \in \P_\textsf{config1}$.
\end{proof}

We use the following Algorithm~\ref{algo:frknapcenterround1} to round any $(x,y,q) \in \P_\textsf{config1}$:

\begin{algorithm}[h]
\caption{$\textsc{FRKnapCenterRound1}\left(x,y,q \right)$}
\begin{algorithmic}[1]
\STATE Randomly pick a set $U \in \U$ with probability $q_U$
\STATE Let $x'_{ij} \gets x^{U}_{ij}/q_U$ and $y'_i \gets \min\{y_i^U/q_U, 1\}$
\RETURN $\S = \textsc{BasicRFKnapCenterRound}\left(x', y' \right)$
\end{algorithmic} 
\label{algo:frknapcenterround1}
\end{algorithm}

We are now ready to prove Theorem \ref{thm:FRKnapCenter}.

\begin{proof}[Proof of Theorem \ref{thm:FRKnapCenter}]
We will show that solution $\S$ returned by \textsc{FRKnapCenterRound1} satisfies the requirements of Theorem \ref{thm:FRKnapCenter}. Let $\e(U)$ denote the event that it select $U \in \U$. Note that $(x', y')$ satisfies the following constraints:
\begin{align*}
	\sum_{j \in V} \sum_{i \in B_j} x_{ij}' &\geq t, \\
	\sum_{i \in B_j} x_{ij}' &\leq 1, \quad \forall j \in V, \\
	\sum_{i \in B_j} x_{ij}' &= \sum_{i \in B_j} x_{ij} / q_U, \quad  \forall j \in V, \\
	x_{ij}' &\leq y_i', \quad \forall i,j \in V, \\
	\sum_{i \in V} w_iy_i' &\leq 1.
\end{align*}

Moreover, $y'_i = 1$ for all $i \in U$, and $y'_i = 0$ for all big vertices $i \in V \setminus U$. Thus, the two extra fractional vertices opened by \textsc{BasicFRKnapCenterRound} have  weight at most $\epsilon$. By Proposition \ref{prop:basicrfknapcenter}, we have $w(\S) \leq 1 + 2\epsilon$. Moreover, conditioned on $U$, we have 
$$\Pr[j \in \C | \e(U)] \geq  \sum_{i \in B_j} x_{ij}' = \sum_{i \in B_j} x_{ij} / q_U.$$
Thus, by definition of $\P_\textsf{config1}$ and our construction of $\S$, we get
\begin{align*}
	\Pr[j \in \C] &= \sum_{U \in \U} \Pr[j \in \C | \e(U)] \Pr[\e(U)] \geq \sum_{U \in \U} \sum_{i \in B_j} x_{ij} \geq p_j.
\end{align*}

\end{proof}

\subsubsection{An algorithm that satisfies the knapsack constraint exactly} \label{sec:knap_exact}

Let $\epsilon > 0$ a small parameter to be determined. Let $\U$ denote the collection of all possible vertex sets of at most $\lceil 1/\epsilon \rceil$; note that $|\U| \leq n^{O(1/\epsilon)}$. Suppose $R$ is the optimal radius to our instance. Given a set $U \in \U$, we say that vertex $j \in V$ is \emph{blue} if $d(j,U) \leq 3 R$; otherwise, vertex $i$ is \emph{red}. For any $i \in V$, let $\RBall(i, U, R)$ denote the set of red vertices within radius $3R$ from $i$:
$$ \RBall(i, U, R) := \{j \in V: d(i,j) \leq 3R ~\wedge~  d(j,U) > 3 R \}
$$

Consider the \emph{configuration} polytope $\P_\textsf{config2}$ containing points $(x,y,q)$ with the following constraints:
\begin{align*}
  \begin{cases}
      \sum_{U \in \U}q_U = 1     & \quad \\
      \sum_{i \in B_j} x^{U}_{ij} \leq q_U    & \quad \forall j \in V, U \in \U \\
      \sum_{U \in \U} \sum_{i \in B_j} x^{U}_{ij} \geq p_j    & \quad \forall j \in V \\
       x^{U}_{ij} \leq y^{U}_i   & \quad \forall i,j \in V, U \in \U\\
       \sum_{i \in V} w_i y^{U}_i \leq q_U B   & \quad  \forall U \in \U \\
       \sum_{j \in V} \sum_{i \in B_j} x^{U}_{ij} \geq q_U t   & \quad \\
       y^{U}_i = 1 & \quad \forall U \in \U, i \in U\\
       y^{U}_i = 0 & \quad \forall U \in \U, i \in V \setminus U, |\RBall(i, U, R)| \geq \epsilon n \\
       x^{U}_{ij}, y^{U}_i, q_U \geq 0   & \quad \forall i,j \in V, U \in \U \\
  \end{cases}
\end{align*}

We first claim that $\P_\textsf{config2}$ is a valid relaxation polytope for the problem.

\begin{proposition} The polytope $\P_\textsf{config2}$ is non-empty.
\label{prop:valid_knap}
\end{proposition}
\begin{proof}
Suppose $\S$ is a solution drawn from the optimal distribution $\D$. Consider the  following randomized procedure to generate a random subset $U_\S$ of $\S$:
\begin{algorithm}[h]
\caption{Generate random $U_\S$:}
\begin{algorithmic}[1]
\STATE Set $U_\S \leftarrow \emptyset$.
\WHILE{there exists $\in S$ such that $|\RBall(i, U_\S, R)| \geq \epsilon n$}
\STATE Select the vertex $i$ of smallest index such that $|\RBall(i, U_\S, R)| \geq \epsilon n$
\STATE Update $U_\S \leftarrow U_\S \cup \{i \}$
\STATE Mark all vertices within radius $3R$ of $i$ as blue
\ENDWHILE
\end{algorithmic}
\label{algo:genus}
\end{algorithm}

Note that for all $i \in \S \setminus U_\S$, we have $ |\RBall(i, U_\S, R)| < \epsilon n$ by the condition of the while-loop. Moreover, we claim that $|U_\S| \leq \lceil 1/\epsilon \rceil$, so that $U_\S \in \U$. For, suppose $|U_\S| > 1/\epsilon$; for each $i \in U_\S$, there are at least $\epsilon n$ red vertices turned into blue by $i$ in the procedure. This implies that there are more than $(1/\epsilon) \times \epsilon n = n$ vertices, which is a contradiction.

Now for any $U \in \U$, we set $q_U := \Pr[U_\S = U]$. Let $x^{U}_{ij} = \Pr[U_S = \U \wedge \text{$j$ is connected to $i$}]$, and finally let $y^{U}_i = \Pr[ U_\S = U \wedge i \in \S]$. Then it is clear that $\sum_{U \in \U}q_U = 1$. Using similar arguments as in the proof of Proposition \ref{prop:relaxation_knapcenter}, we have the following inequalities:
\begin{align*}
\sum_{i \in B_j} x^{U}_{ij} \leq q_U,    & \quad \forall j \in V, U \in \U \\
      \sum_{U \in \U} \sum_{i \in B_j} x^{U}_{ij} \geq p_j,    & \quad \forall j \in V \\
      \sum_{i \in V} w_i y^{U}_i \leq q_U,   & \quad  \forall U \in \U \\
      \sum_{j \in V} \sum_{i \in B_j} x^{U}_{ij} \geq q_U t.   & \quad 
\end{align*}
As mentioned before, if $|\RBall(i, U_\S, R)| \geq \epsilon n$ then $i \notin \S$. Therefore,
\begin{align*}
  y^{U}_i = \Pr[U_\S = U \wedge i \in \S] = 0, & \quad \forall U \in \U, i \in V \setminus U, |\RBall(i, U, R)| \geq \epsilon n.
\end{align*}
The other constraints can be verified easily. We conclude that $(x,y,q) \in \P_\textsf{config2}$.
\end{proof}

Next, let us pick any $(x,y,q) \in \P_\textsf{config2}$ and use the following algorithm to round it. 

\begin{algorithm}[h]
\caption{$\textsc{FRKnapCenterRound2}\left(x,y,q \right)$}
\begin{algorithmic}[1]
\STATE Randomly pick a set $U \in \U$ with probability $q_U$
\STATE Let $x'_{ij} \gets x^{U}_{ij}/q_U$ and $y'_i \gets \min\{y_i^U/q_U, 1\}$
\STATE $\S' \gets \textsc{BasicRFKnapCenterRound}\left(x', y' \right)$
\STATE Let $i_1, i_2$ be vertices in $\S' \setminus U$ having largest weights.
\RETURN $\S = \S' \setminus \{i_1, i_2\}$
\end{algorithmic} 
\end{algorithm}

\bigskip \noindent \textbf{Analysis.} Let us fix any $\gamma > 0$ and set $\epsilon := \frac{\gamma^2}{2}$. Also, let $\e(U)$ denote the event that $U \in \U$ is picked in the algorithm. Again, observe that $(x', y')$ satisfies the following inequalities:

\begin{align*}
	\sum_{j \in V} \sum_{i \in B_j} x_{ij}' &\geq t, \\
	\sum_{i \in B_j} x_{ij}' &\leq 1, \quad \forall j \in V, \\
	\sum_{i \in B_j} x_{ij}' &= \sum_{i \in B_j} x_{ij} / q_U, \quad  \forall j \in V, \\
	x_{ij}' &\leq y_i', \quad \forall i,j \in V, \\
	\sum_{i \in V} w_iy_i' &\leq 1.
\end{align*}

Recall that the algorithm \textsc{BasicFRKnapCenterRound} will return a solution $\S'$ consisting of a set $\S''$ with $w(\S'') \leq 1$ plus (at most) two extra ``fractional'' centers $i^*$ and $i^{**}$. Since $y'_{i^{**}}, y'_{i^*}$ are  fractional, we have that $i^*, i^{**} \notin U$. Thus, by removing the two centers having highest weights in $\S' \setminus U$, we ensure that $w(\S) \leq 1$ with probability one.

Now we shall prove the coverage guarantee. By Proposition \ref{prop:basicrfknapcenter}, $\S'$ covers at least $t$ vertices within radius $3R$. If a vertex is blue, it can always be connected to some center in $U$ and so it is not affected by the removal of $i_1, i_2$. Because each of $i_1$ and $i_2$ can cover at most $\epsilon n$ other red vertices, we have 
$$|\C| \geq t - 2\epsilon n = 1 - \gamma^2 n.$$

For any $j \in V$, let $X_j$ be the indicator random variable for the event that $d(j, \S') \leq 3 R$ but $d(j, \S \setminus \{i_1, i_2 \}) > 3 R$. We say that $j$ is a bad vertex iff $\E[X_j] \geq \gamma$, other, vertex $j$ is good. Note that $\sum_{j \in V} X_j \leq 2\epsilon n$ with probability one. Thus, there can be at most $2\epsilon n / \gamma$ bad vertices. Letting $T$ denote the set of good vertices, we have
$$|T| \geq n - 2\epsilon n/\gamma = (1-\gamma)n.$$ 
By Proposition \ref{prop:basicrfknapcenter}, $\Pr[j \text{ is covered by }\S'] \geq p_j$. For any $j \in T$, we have
\begin{align*}
	\Pr[j \in \C] &= \Pr[j \text{ is covered by }\S' \wedge X_j = 0] \\ 
		&= \Pr[j \text{ is covered by }\S'] - \Pr[j \text{ is covered by }\S' \wedge X_j = 1] \\ 
		&\geq \Pr[j \text{ is covered by }\S'] - \Pr[X_j = 1] \geq p_j - \gamma.
\end{align*}
This concludes the first part of Theorem \ref{thm:FRMatCenter} for the \textsf{FRKnapCenter} problem.

\section{The Matroid Center problems with outliers}
In this section, we will first give a tight $3$-approximation algorithm for the \textsf{RMatCenter} problem, improving upon the $7$-approximation algorithm by Chen et. al. \cite{jianli_kc}. Then we study the \textsf{FRMatCenter} problem and prove the second part of Theorem \ref{thm:FRMatCenter}.

\subsection{The robust matroid center problem}
Suppose $\I = (V, d, \M, t)$ is an instance of the \textsf{RMatCenter} problem with optimal radius $R$. Let $r_\M$ denote the rank function of $\M$. Consider the polytope $\P_\textsf{RMatCenter}$ containing points $(x, y)$ satisfying constraints (\ref{ineq:cov})--(\ref{ineq:nonneg}), and the matroid rank constraints:
\begin{align}
	 y(U) \leq r_\M(U), \quad \forall U \subseteq V. \label{ineq:matroid}
\end{align}

Since $R$ is the optimal radius, it is not difficult to check that $\P_\textsf{RMatCenter} \neq \emptyset$. Let us pick any fractional solution $(x,y) \in \P_\textsf{RMatCenter}$. The next step is to round $(x,y)$ into an integral solution. Our $3$-approximation algorithm is as follows.

\begin{algorithm}[h]
\caption{$\textsc{RMatCenterRound}\left(x,y \right)$}
\begin{algorithmic}[1]
\STATE $(V',\vec{c}) \gets \textsc{RFiltering}\left(x,y \right).$
\STATE Let $\P' := \left\{z \in [0,1]^V: z(U) \leq r_\M(U) ~\forall U \subseteq V ~\wedge~ z(F_i) \leq 1 ~~\forall i \in V' \right\}$
\STATE Find a basic solution $Y \in \P'$ which maximizes the linear function $f: [0,1]^V \to \mathbb{R}$ defined as
	$$ f(z) := \sum_{j \in V'} c_j \sum_{i \in F_j} z_i ~~\text{for }z \in [0,1]^V.$$
\RETURN $\S = \{i \in V: Y_i = 1\}$.
\end{algorithmic} 
\label{algo:rmcenter}
\end{algorithm}

\bigskip \noindent \textbf{Analysis.} Again, by construction, the clusters $F_i$ are pairwise disjoint for $i \in V'$. Note that $\P'$ is the matroid intersection polytope between $\M$ and another partition matroid polytope saying that at most one item per set $F_i$ for $i \in V'$ can be chosen. Moreover, $y \in \P'$ implies that $\P' \neq \emptyset$. Thus, $\P'$ has integral extreme points and optimizing over $\P'$ can be done in polynomial time. Note that the solution $\S$ is feasible as it satisfies the matroid constraint. The correctness of $\textsc{RMatCenterRound}$ follows immediately by the following two propositions. 

\begin{proposition} There are at least $f(Y)$ vertices in $V$ that are at distance at most $3R$ from some open center in $\S$.
\label{prop:fbound}
\end{proposition}
\begin{proof} Recall that $\S$ is the set of vertices $i \in V$ such that $Y_i = 1$. Moreover, since $Y(F_j) \leq 1$, there can be at most one open center in $F_j$ (i.e., $|\S \cap F_j| \leq 1$) for each $j \in V'$. For any $j \in V'$,
\begin{itemize}
	\item if $Y(F_j) = 0$, then there is no open center in $F_j$ and its contribution in $f(Y)$ is zero,
	\item if $Y(F_j) = 1$, then we open some center $i \in F_j$ and the contribution of $j$ to $f(Y)$ is equal to $c_j$. (Recall that $c_j$ is the number of clusters $F_k$ with $F_j \cap F_k \neq \emptyset$.) By triangle inequality, $d(k,i) \leq d(k,j) + d(j,i) \leq 2R+R = 3R$.
\end{itemize}
\end{proof}

\begin{proposition} We have $f(Y) \geq t$.
\label{prop:fbound2}
\end{proposition}
\begin{proof}
For each $j \in V'$ and $i \in F_j$, define $y'_i := x_{ij}$ (this is well-defined as all clusters $F_j$ for $j \in V'$ are pairwise disjoint). Also, set $y'_i := 0$ for other vertices $i$ not belonging to any marked cluster. Then, by greedy choice and constraint (\ref{ineq:cov}), we have
\begin{align*}
	f(y') &= \sum_{j \in V'} c_j y'(F_j) = \sum_{j \in V'} c_j s_j \geq \sum_{j \in V}  s_j \geq t.
\end{align*}
By the choice of $Y$, we have $f(Y) \geq f(y') \geq t$.
\end{proof}
This analysis proves the second part of Theorem \ref{theorem:standard_models}.

\subsection{The fair robust matroid center problem}

In this section, we consider the \textsf{FRMatCenter} problem. It is not difficult to modify and randomize algorithm \textsc{RMCenterRound} so that it would return a random solution satisfying both the fairness guarantee and matroid constraint, and preserving the coverage constraint \emph{in expectation}. This can be done by randomly picking $Y$ inside $\P'$. However, if we want to obtain some concrete guarantee on the coverage constraint, we may have to (slightly) violate either the matroid constraint or the fairness guarantee. We leave it as an open question whether there exists a true approximation algorithm for this problem.

We will start with a pseudo-approximation algorithm which always returns a basis of $\M$ plus at most one extra center.  Our algorithm is quite involved. We first carefully round a fractional solution inside a matroid intersection polytope into a (random) point with a special property: the unrounded variables form a single path connecting some clusters and tight matroid rank constraints. Next, rounding this point will ensure that all but one cluster has an open center. Then opening one extra center is sufficient to cover at least $t$ clients.

Finally, using a similar preprocessing step similar to the one in Section \ref{sec:knap_exact}, we correct the solution by removing the extra center without affecting the fairness and coverage guarantees by too much. This algorithm concludes Theorem \ref{thm:FRMatCenter}.

\subsubsection{A pseudo-approximation algorithm}

Suppose $\I = (V, d, \M, t, \vec{p})$ is an instance the robust matroid center problem with the optimal radius $R$. Let $r_\M$ denote the rank function of $\M$ and $\P_\M$ be the matroid base polytope of $\M$. Consider the polytope $\P_\textsf{FRMatCenter}$  containing points $(x, y)$ satisfying constraints (\ref{ineq:cov})--(\ref{ineq:nonneg}), the fairness constraint (\ref{ineq:fair}), and the matroid constraints (\ref{ineq:matroid}).  We note that $\P_\textsf{FRMatCenter}$ is a valid relaxation.

The main algorithm is summarized in Algorithm \ref{algo:frmcenter}, which can round any \textit{vertex} point $(x, y) \in \P_\textsf{FRMatCenter}$. Basically, we will round $y$ iteratively. In each round, we construct a (multi)-bipartite graph where vertices on the left side are the disjoint sets $O_1, O_2, \ldots$ in Corollary \ref{cor:decompose_face}. Vertices on the right side are corresponding to the disjoint sets $F_1, F_2, \ldots$ returned by \textsc{RFiltering}. Now each edge of the bipartite graph, connecting $O_i$ and $F_j$, represents some unrounded variable $y_v \in (0,1)$ where $v \in O_i$ and $v \in F_j$. See Figure \ref{fig:bipartite_graph}.

\begin{figure}[h]
\centering
\includegraphics[width=0.3\textwidth]{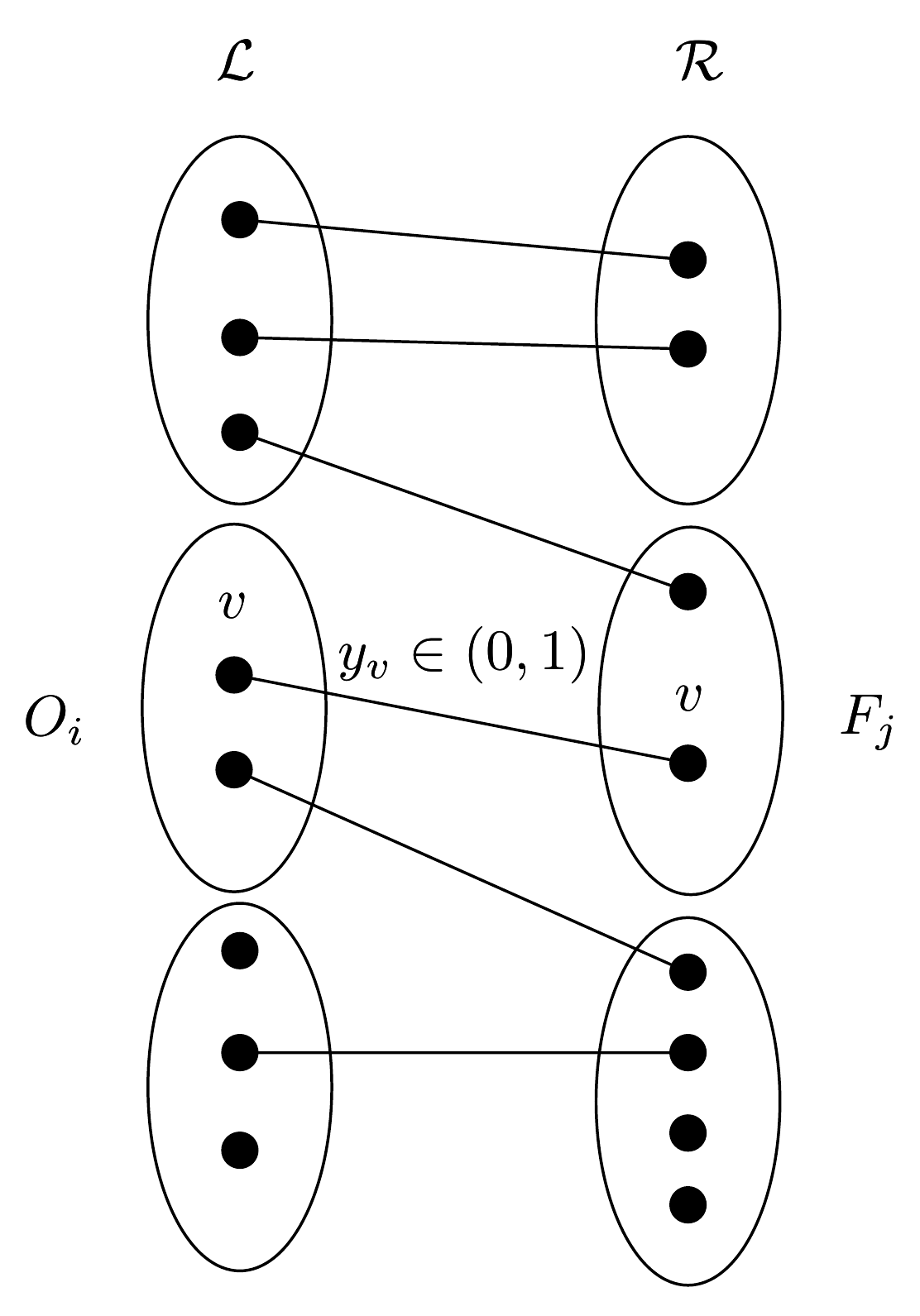}
\caption{Construction of the multi-bipartite graph $H=(\L,\R,E_H)$ in the main algorithm. }
\label{fig:bipartite_graph}
\end{figure}

Then we carefully pick a cycle (path) on this graph and round variables on the edges of this cycle (path). This is done by subroutines \textsc{RoundCycle}, \textsc{RoundSinglePath}, and \textsc{RoundTwoPaths}. See Figures \ref{fig:cycle}, \ref{fig:single_path}, and \ref{fig:two_paths}. Basically, these procedures will first choose a direction $\vec{r}$ which alternatively increases and decreases the variables on the cycle (path) so that (i) all tight matroid constraints are preserved and (ii) the number of (fractionally) covered clients is also preserved. Now we randomly move $y$ along $\vec{r}$ or $-\vec{r}$ to ensure that all the marginal probabilities are preserved.

Finally, all the remaining, fractional variables will form one path on the bipartite graph. We round these variables by the procedure \textsc{RoundFinalPath} which exploits the integrality of any face of a matroid intersection polytope. Then, to cover at least $t$ clients, we may need to open one extra facility.

\begin{algorithm}[H]
\caption{$\textsc{PseudoFRMCenterRound}\left(x,y \right)$}
\begin{algorithmic}[1]
\STATE $(V',\vec{c}) \gets \textsc{RFiltering}\left(x,y \right)$ and let $\F \gets \{F_j: j \in V'\}$

\STATE Set $y'_i \gets x_{ij}$ for all $j \in V', i \in F_j$
\STATE Set $y'_i \gets 0$ for all $i \in V \setminus \bigcup_{j \in V'}F_j$

\WHILE{$y'$ still contains some fractional values}
	\STATE Note that $y'\in \P_\M$. Compute the disjoint sets $O_1, \ldots, O_t$ and constants $b_{O_1}, \ldots, b_{O_t}$ as in Corollary \ref{cor:decompose_face}. 
	\STATE Let $O_0 \gets V \setminus \bigcup_{i=1}^t O_i$ and $F_0 \gets V \setminus \bigcup_{j \in V'}F_j  $
	\STATE Construct a multi-bipartite graph $H = (\L, \R, E_H)$ where
	\begin{itemize}[noitemsep,nolistsep] 
		\item each vertex $i \in \L$, where $\L = \{0,\ldots,t\}$, is corresponding to the set $O_i$
		\item each vertex $j \in \R$, where $\R = \{0\} \cup \{k:F_k \in \F\}$, is corresponding to the set $F_j$ 
		\item for each vertex $v \in V$ such that $y_v \in (0,1)$: if $v$ belongs to some set $O_i$ and $F_j$, add an edge $e$ with label $v$ connecting $i \in \L$ and $j \in \R$. 
	\end{itemize}
	\STATE Check the following cases (in order):
	\begin{itemize}[noitemsep,nolistsep] 
		\item Case 1: $H$ contains a cycle. Let $\vec{v} = (v_1, v_2, \ldots, v_{2\ell})$ be the sequence of edge labels on this cycle. Update $y' \gets \textsc{ RoundCycle}(y', \vec{v})$ and go to line 4.
		\item Case 2: $H$ contains a maximal path with one endpoint in $\L$ and the other in $\R$. Let $\vec{v} = (v_1, v_2, \ldots, v_{2\ell+1})$ be the sequence of edge labels on this path. Update $y' \gets \textsc{ RoundSinglePath}(y', \vec{v})$ and go to line 4.
		\item Case 3: There are at least 2 distinct maximal paths (not necessarily disjoint) having both endpoints in $\R$. Let $\vec{v}_1, \vec{v}_2$ be the sequences of edge labels on these two paths. Update $y' \gets \textsc{ RoundTwoPaths}(y', \vec{v}_1, \vec{v}_2, \vec{c})$ and go to line 4.
		\item The remaining case: all edges in $H$ form a single path with both endpoints in $\R$. Let $(v_1, v_2, \ldots, v_{2\ell})$ be the sequence of edge labels on this path. Let $Y \gets \textsc{ RoundFinalPath}(y', \vec{v})$ and exit the loop.
	\end{itemize}
\ENDWHILE
\RETURN $\S = \{i \in V: Y_i = 1\}$.
\end{algorithmic} 
\label{algo:frmcenter}
\end{algorithm}

\begin{proposition} The polytyope $\P_\textsf{FRMatCenter}$ is non-empty.
\end{proposition}
\begin{proof}
The proof is similar to Proposition \ref{prop:frkcenter}, and is omitted.
\end{proof}

We now explain how to implement the rounding steps in this algorithm. These are all based on the subroutine \textsc{RoundSinglePoint}, which moves the vector $y$ as far as possible along the direction of vector $\vec r$ until it a new face of the polytope begins tight. We define it formally as Algorithm~\ref{roundsingle}; note that it can be implemented efficiently by solving an appropriate LP.

\begin{algorithm}[H]
\caption{$\textsc{RoundSinglePoint}\left(y, \vec{r} \right)$}
\begin{algorithmic}[1]
\STATE Select $\delta^* \geq 0$ to be maximal such that the vector $z + \delta^* r$ is in $\P_{\M}$.
\RETURN $(y + \delta^* \vec{r}, \delta^*)$.
\end{algorithmic} 
\label{roundsingle}
\end{algorithm}

\begin{algorithm}[H]
\caption{$\textsc{RoundCycle}\left(y', \vec{v} \right)$}
\begin{algorithmic}[1]
\STATE Initialize $\vec{r} = \vec{0}$, then set $r_{v_j} = (-1)^j$ for $j = 1,2, \ldots, |\vec{v}|$
\STATE $(y_1, \delta_1) \gets $\textsc{RoundSinglePoint}$(y', \vec{r})$
\RETURN $y_1$
\end{algorithmic} 
\end{algorithm}

\begin{figure}[H]
\centering
\includegraphics[width=0.5\textwidth]{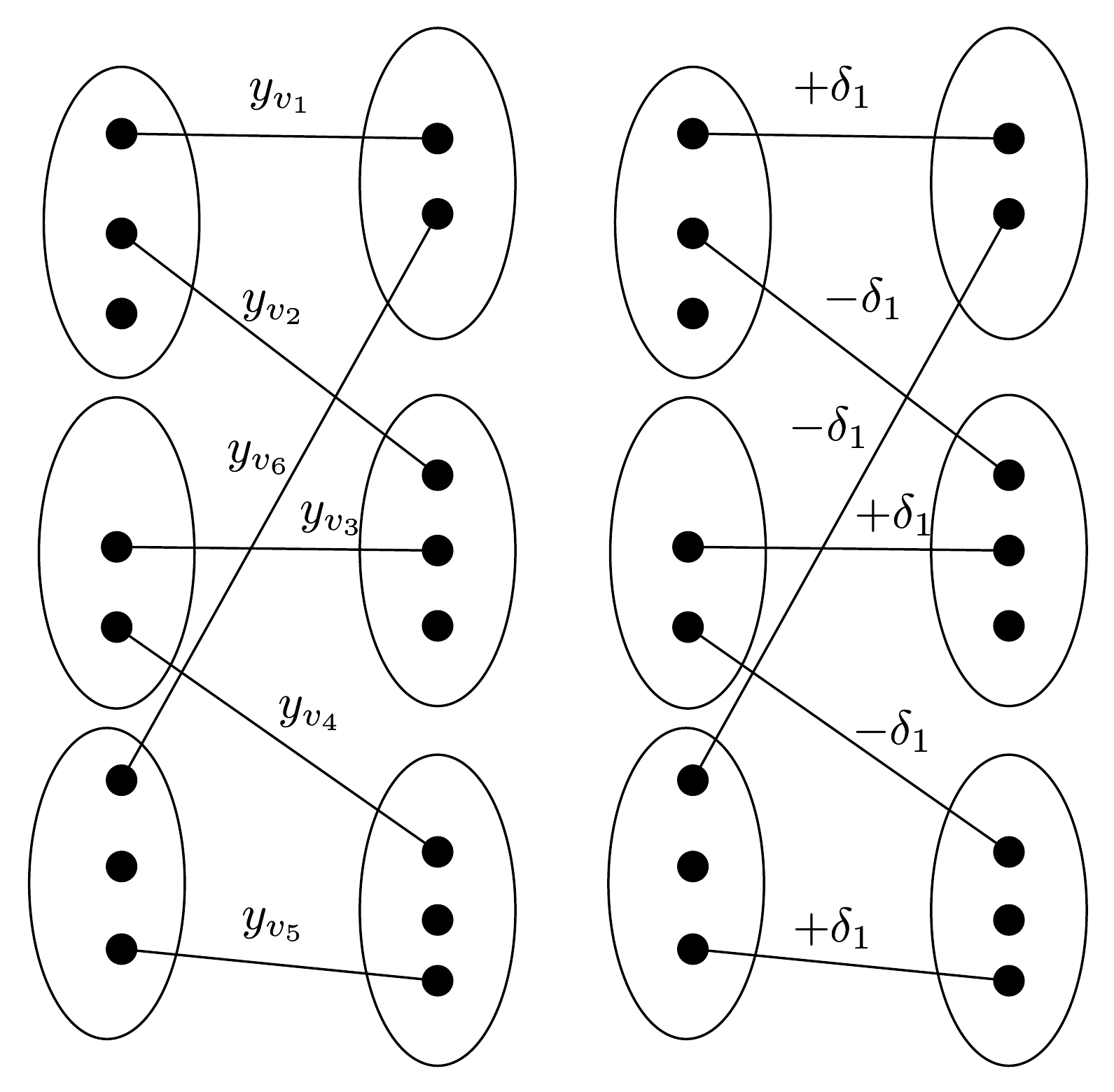}
\caption{The left part shows a cycle. The right part shows how the variables on the cycle are being changed by \textsc{RoundCycle}.}
\label{fig:cycle}
\end{figure}

\begin{algorithm}[H]
\caption{$\textsc{RoundSinglePath}\left(y', \vec{v} \right)$}
\begin{algorithmic}[1]
\STATE Initialize $\vec{r} = \vec{0}$, then set $r_{v_j} = (-1)^{j+1}$ for $j = 1,2, \ldots, |\vec{v}|$
\STATE $(y_1, \delta_1) \gets$\textsc{RoundSinglePoint}$(y', \vec{r})$
\RETURN $y_1$
\end{algorithmic} 
\label{algo:round_single_path}
\end{algorithm}

\begin{figure}[H]
\centering
\includegraphics[width=0.5\textwidth]{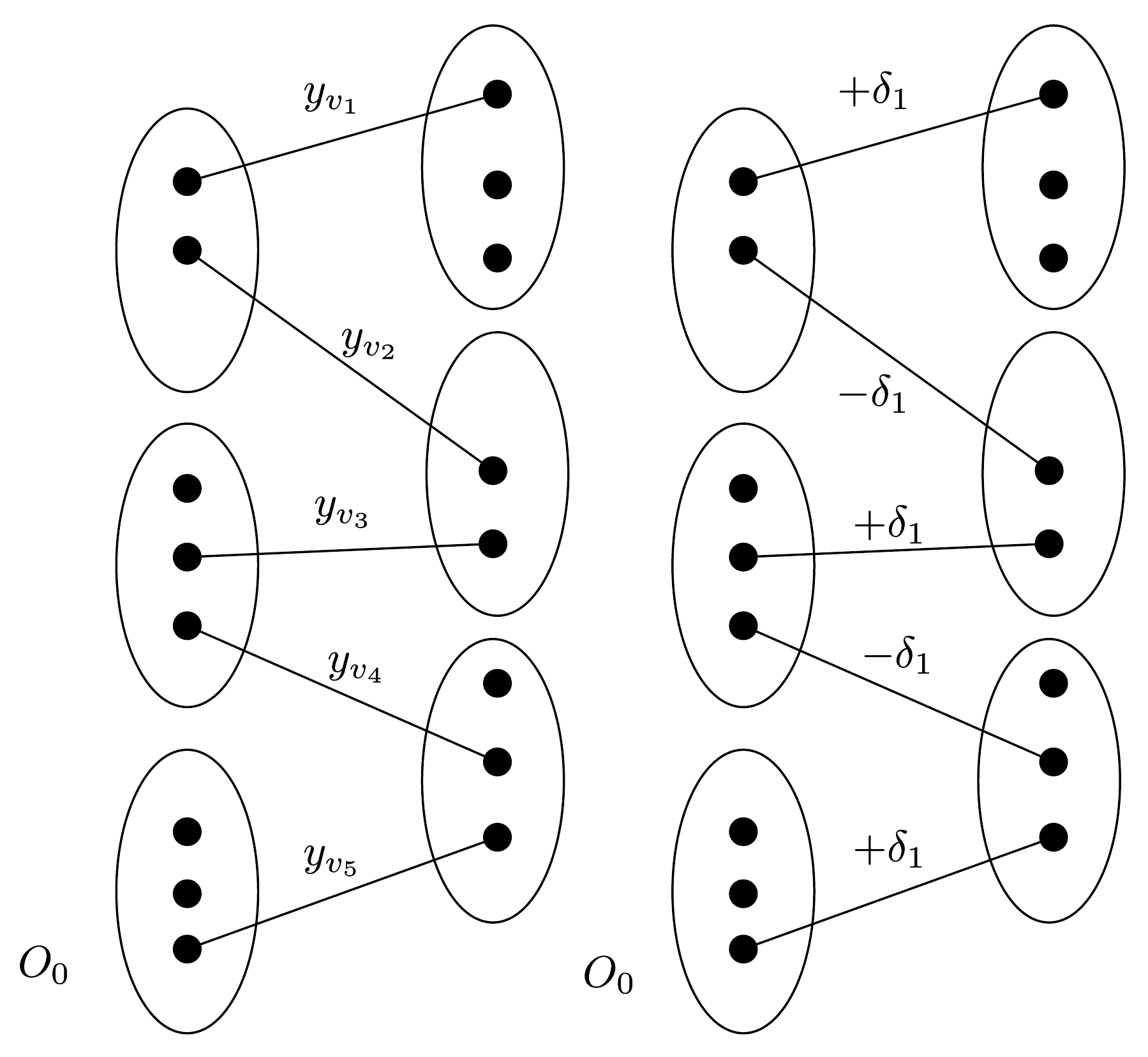}
        \caption{The left part shows a single path. The right part shows how the variables on the path are being changed by \textsc{RoundSinglePath}.}
        \label{fig:single_path}
\end{figure}

\begin{algorithm}[H]
\caption{$\textsc{RoundTwoPaths}\left(y', \vec{v}, \vec{v}', \vec{c} \right)$}
\begin{algorithmic}[1]
\STATE WLOG, suppose $j_1, j_2 \in \R$ are endpoints of $v_1, v_{2\ell}$ of the path $\vec{v}$ respectively and $c_{j_1} \geq c_{j_2}$
\STATE WLOG, suppose $j_1', j_2' \in \R$ are endpoints of $v_1', v_{2\ell'}'$ of the path $\vec{v}'$ respectively and $c_{j_1'} \geq c_{j_2'}$
\STATE $\Delta_1 \gets c_{j_1} - c_{j_2};~~~ \Delta_2  \gets c_{j_1'} - c_{j_2'};~~~ \vec{r} \gets \vec{0}$ 
\STATE $V_1^+ \gets \{v_1, v_3, \ldots, v_{2\ell-1} \}; V_1^- \gets \{v_2, v_4, \ldots, v_{2\ell}\}$
\STATE $V_2^+ \gets \{v_2', v_4', \ldots, v_{2\ell'}'\}; V_2^- \gets \{v_1', v_3', \ldots, v_{2\ell'-1}' \}$
\STATE \textbf{for each} $v \in V_1^+$: $r_v \gets r_v + 1$; \textbf{for each} $v \in V_1^-$: $r_v \gets r_v - 1$
\STATE \textbf{for each} $v \in V_2^+$: $r_v \gets r_v + \Delta_1/\Delta_2$; \textbf{for each} $v \in V_2^-$: $r_v \gets r_v - \Delta_1/\Delta_2$
\STATE $(y_1, \delta_1) \gets $\textsc{RoundSinglePoint}$(y', \vec{r})$
\STATE $(y_2, \delta_2) \gets $\textsc{RoundSinglePoint}$(y', -\vec{r})$
\STATE With probability $\delta_1 / (\delta_1+\delta_2)$: \textbf{return} $y_2$
\STATE With remaining probability $\delta_2 / (\delta_1+\delta_2)$: \textbf{return} $y_1$
\end{algorithmic} 
\label{algo:round_two_paths}
\end{algorithm}

\begin{figure}[H]
\centering
\includegraphics[width=1\textwidth]{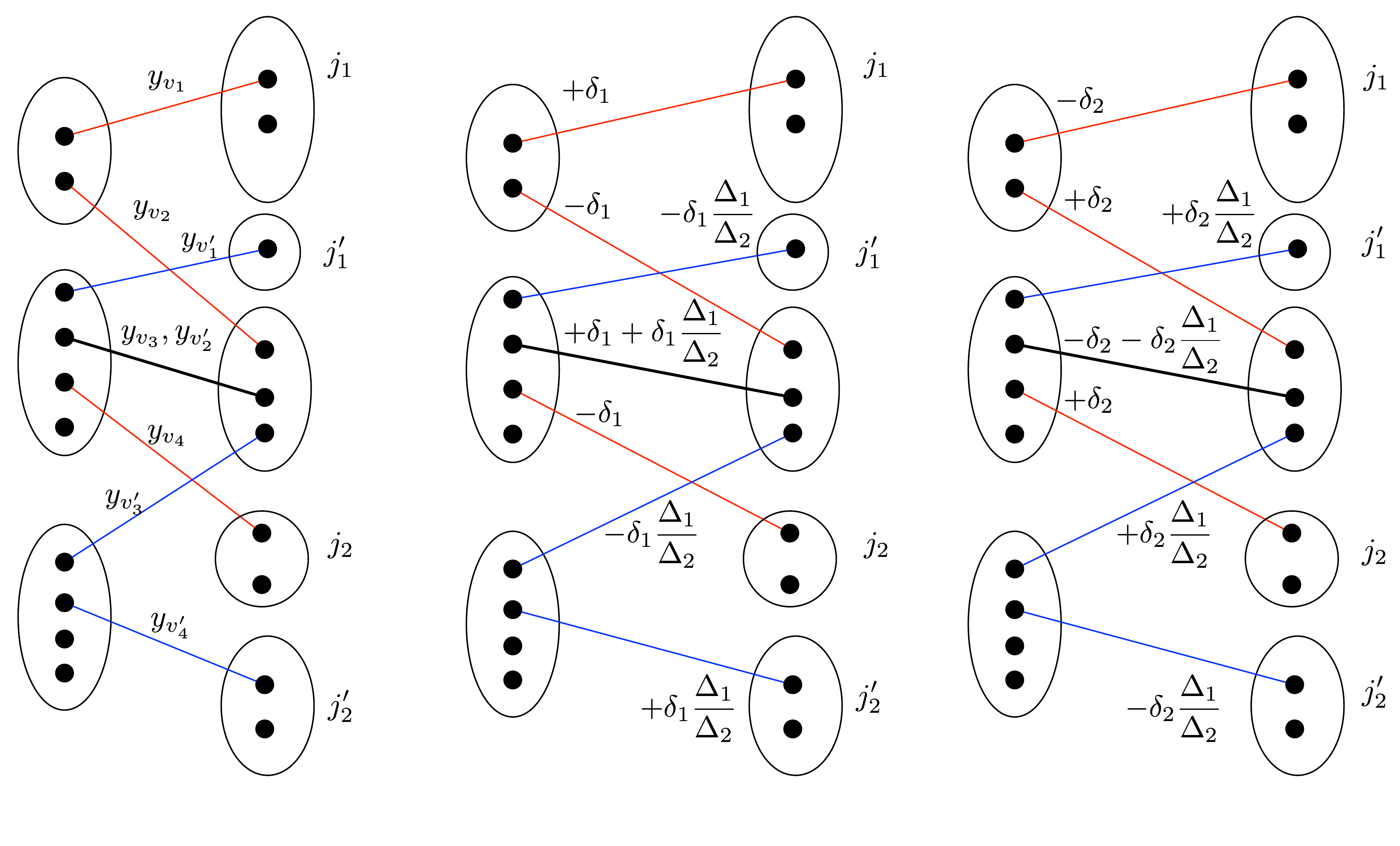}
        \caption{The left part shows an example of two distinct maximal paths chosen in Case 3. The black edge is common in both paths. The middle and right parts are two possibilities of rounding $y$. With probability $\delta_1/(\delta_1+\delta_2)$, the strategy in the right part is adopted. Otherwise, the strategy in the middle part is chosen.}
        \label{fig:two_paths}
\end{figure}

\begin{algorithm}[H]
\caption{$\textsc{RoundFinalPath}\left(y, \vec{v} \right)$}
\begin{algorithmic}[1]
\STATE $\P_1 \gets  \left\{z \in [0,1]^V: z(U) \leq r_\M(U) ~\forall U \subseteq V \wedge z(O_i) = b_{O_i}  ~\forall i \in \L\setminus\{0\} \wedge z_i=0 ~\forall i:y_i=0 \right\}$
\STATE $\P_2 \gets \{z \in [0,1]^V:  z(F_j) = y(F_j) ~\forall j \in V'\setminus J \wedge z(F_j) \leq 1 ~\forall j \in J \}$, where $J \subseteq \R$ is the set of vertices in $\R$ on the path $\vec{v}$.

\STATE Pick an arbitrary extreme point $\hat{y}$ of $\P' = \P_1 \cap \P_2$
\STATE \textbf{for each} $j \in \R$ and $j$ is on the path $\vec{v}$: if $\hat{y}(F_j) = 0$, pick an arbitrary $u \in F_j$ and set $\hat{y}_u \gets 1$.
\RETURN $\hat{y}$
\end{algorithmic} 
\label{algo:round_final_path}
\end{algorithm}

\subsubsection{Analysis of \textsc{PseudoFRMCenterRound}}

\begin{proposition} In all but the last iteration, the while-loop (lines $4$ to $8$) of \textsc{PseudoFRMCenterRound} preserves the following invariant: if $y'$ lies in the face $D$ of $\P_\M$ (w.r.t all tight matroid rank constraints) at the beginning of the iteration, then $y' \in D$ at the end of this iteration.
\label{prop:invariant}
\end{proposition}
\begin{proof}
Observe that $y' \in \P_\M$ at the beginning of the first iteration due to the definition of $y'$. Fix any iteration. Let $y''$ be the updated $y'$ at the end of the iteration. By Corollary \ref{cor:decompose_face}, it suffices to show that
$$ y'' \in \left\{x \in \mathbb{R}^{n}: x(S) = b_S~~\forall S \in \O; ~~~ x_i = 0 ~~\forall i \in J; ~~~ x \in \P_\M  \right\}, $$
where $J \subseteq V$ is the set of all vertices $i$ with $y'_i = 0$. Note that $y''$ is the output of one of the three subroutines \textsc{RoundCycle}, \textsc{RoundSinglePath}, and \textsc{RoundTwoPaths}. Since we only round floating variables strictly greater than zero, we have  $y''_i = 0$ for all $i \in J$. Also, the procedure \textsc{RoundSinglePoint} guarantees that $y'' \in \P_\M$.
\begin{itemize}
	\item When calling the procedure \textsc{RoundCycle}, observe that each vertex $j \in \L$ on the cycle is adjacent to exactly two edges. By construction, we always increase the variable on one edge and decrease the variable on the other edge at the same rate. See Figure \ref{fig:cycle}. Therefore, $y''(O_j) = y'(O_j) = b_{O_j}$ for all $j \in \O$.
	\item When calling the procedure \textsc{RoundSinglePath}, recall that our path is maximal and has one endpoint in $\L$ and the other in $\R$. We claim that the left endpoint of this path should be corresponding to the set $O_0$. Otherwise, suppose it is some set $O_j$ with $j>0$. We have the tight constraint $y'(O_j) = b_{O_j} \in \mathbb{Z}^+$. Then the degree of the vertex $j$ must be at least $2$ as there must be at least two fractional variables in this set. This contradicts to the fact that our path is maximal. See Figure \ref{fig:single_path}. By the same argument as before, we have $y''(O_j) = y'(O_j) = b_{O_j}$ for all $j \in \O$.
	\item In the procedure \textsc{RoundTwoPaths}, we round the variables on two paths which have both endpoints in $\R$. Thus, any vertex $j$ should be adjacent to either $2$ or $4$ edges. Again, by construction, the net change in $y'(O_j)$ is equal to zero. See Figure \ref{fig:two_paths}.
\end{itemize}
Finally, the claim follows by induction.
\end{proof}

\begin{proposition} \textsc{PseudoFRMCenterRound} terminates in polynomial time.
\end{proposition}
\begin{proof}
Note that, in each iteration, each floating variable $y'_v \in (0,1)$ corresponds to exactly one edge in the bipartite graph. This is because, by construction, the sets $O_0, \ldots, O_t$ form a partition of $V$ and the sets in $\F$ and $F_0$ also form a partition of $V$. Thus, as long as there are fractional values in $y'$, our graph will have some cycle or path.

Now we will show that the while-loop (lines $4$ to $8$) terminates after $O(|V|)$ iterations. For any set $S$, let $\chi(S)$ denote the characteristic vector of $S$. That is, $\chi(v) = 1$ for $v \in S$ and $\chi(v) = 0$ otherwise. Let us fix any iteration and let $\T = \{\chi(S): S  \subseteq V ~\wedge~ y'(S) = r_\M(S)\}$ be the set of all tight constraints. In this iteration, we will move $y'$ along some direction $\vec{r}$ as far as possible (by procedure \textsc{RoundSinglePoint}). It means that the new point $y'' = y' + \delta^* \vec{r}$ will either have at least one more rounded variable or hit a new tight constraint $y''(S_0) = r_\M(S_0)$ (while $y'(S_0) < r_\M(S_0)$) for some $S_0 \subseteq V$. Indeed $\chi(S_0)$ is linearly independent of all vectors in $\T$. 

Proposition \ref{prop:invariant} says that all the tight constraints are preserved in the rounding process. Therefore, in the next iteration, we either have at least one more rounded variable or the rank of $\T$ is increased by at least $1$. This implies the algorithm terminates after at most $|V|$ iterations.
\end{proof}

\begin{proposition} In all  iterations, the while-loop (lines $4$ to $8$) of \textsc{PseudoFRMCenterRound} satisfies the invariant that $y'(F_j) \leq 1$ for all $F_j \in \F$.
\label{prop:invariant2}
\end{proposition}
\begin{proof}
By constraints \ref{ineq:connect} and \ref{ineq:open}, this property is true at the beginning of the first iteration. By a very similar argument as in the proof of Proposition \ref{prop:invariant}, this is also true during all but the last iteration. (Note that if $j$ is an endpoint of a path, then $j$ must be adjacent to exactly one fractional value $y_v'$, which could be rounded to one, while other variables $\{y'_{v'}: v' \in F_j, v' \neq v\}$ are already rounded to zero as our path is maximal.) Finally, it is not hard to check that procedure \textsc{RoundFinalPath} also does not violate this invariant.
\end{proof}

\begin{proposition} \textsc{PseudoFRMCenterRound} returns a solution $\S$ which is some independent set of $\M$ plus (at most) one extra vertex in $V$.
\label{prop:feasibility}
\end{proposition}
\begin{proof}
Let us focus on the procedure \textsc{RoundFinalPath}. Recall that the polytope $\P'$ in \textsc{RoundFinalPath} is the intersection of the following two polytopes: 
$$\P_1 =  \left\{z \in [0,1]^V: z(U) \leq r_\M(U) ~\forall U \subseteq V \wedge z(O_i) = b_{O_i}  ~\forall i \in \L\setminus\{0\} \wedge z_i=0 ~\forall i:y_i=0 \right\},$$
and
$$\P_2 = \{z \in [0,1]^V:  z(F_j) = y(F_j) ~\forall j \in V'\setminus J \wedge z(F_j) \leq 1 ~\forall j \in J \},$$ 
where $J \subseteq \R$ is the set of vertices in $\R$ on the path $\vec{v}$.

By construction, $\P_1$ is the face of the matroid base polytope $\P_\M$ corresponding to all tight constraints of $y$. It is well-known that $\P_1$ itself is also a matroid base polytope. By Propositions \ref{prop:invariant} and \ref{prop:invariant2}, we have $y \in \P_1$ and $y \in \P_2$. Thus, $y \in \P$ which implies that $\P \neq \emptyset$. Moreover, $\P_2$ is a partition matroid polytope. (Observe that $z(F_j) = y(F_j) \in \{0,1\} ~\forall j \in V' \setminus J$ since all fractional variables are on the path $\vec{v}$.) Therefore, $\P = \P_1 \cap \P_2$ has integral extreme points and the point $\hat{y}$ chosen in line $3$ is integral.

Finally, recall that $\vec{v} = (v_1, v_2, \ldots, v_{2\ell})$ is a simple path with both endpoints in $\R$. The constraints of $\P_1$ and integrality of $b_{O_i}$'s ensure that $\hat{y}_{v_1}+\hat{y}_{v_2}=1, \hat{y}_{v_3}+\hat{y}_{v_4}=1, \ldots, \hat{y}_{v_{2\ell-1}}+\hat{y}_{v_{2\ell}}=1$. In other words, every vertex $i \in \L$ on the path will be ``matched'' with exactly one vertex in $\R$. Thus, there can be at most one vertex $j \in \R$ on the path such that $\hat{y}(F_j) = 0$ in line $4$. Opening $u \in F_j$ adds one extra facility to our solution.
\end{proof}

Recall that $\C$ is the (random) set of all clients within radius $3R$ from some center in $\S$, where $R$ is the optimal radius. The following two propositions will conclude our analysis.

\begin{proposition} $|\C| \geq t$ with probability one.
\end{proposition}
\begin{proof}
Let $f$ denote the function defined in Algorithm \textsc{RMCenterRound} (i.e., $f(z) = \sum_{j \in V'} \sum_{i \in F_i}z_i$ for any $z \in [0,1]^V$.) Using a similar argument as in the proof of Proposition \ref{prop:fbound}, one can easily verify that there are at least $f(Y)$ vertices in $V$ that are within radius $3R$ from some open center in $\S$. Next, it suffices to show that $f(Y) \geq t$.

By definition of $y'$ in lines $2$ and $3$, we have $f(y') \geq t$ (see the proof of Proposition \ref{prop:fbound2}.) We now claim that $f(y')$ is not decreasing after each iteration of the rounding scheme. We check the following cases:
\begin{itemize}
	\item Case $y'$ is rounded by \textsc{RoundCycle}: observe that $y'(F_j)$ is preserved for all $j \in \R$ since $j$ is adjacent to two edges and we increase/decrease the corresponding variables by the same amount. Thus, $f(y')$ is unchanged.
	\item Case $y'$ is rounded by \textsc{RoundSinglePath}: if $j \in \R$ is not the endpoint of the path then $j$ is adjacent to two edges on the path and $y'(F_j)$ is unchanged. If $j$ is the endpoint, then we increase the variable on the adjacent edge; and hence, $y'(F_j)$ will increase. See Figure \ref{fig:single_path}.
	
	\item Case $y'$ is rounded by \textsc{RoundTwoPaths}: again, for any $j \in \R \setminus \{j_1, j_2, j_1', j_2'\}$,  the value of $y'(F_j)$ remains unchanged in the process. We now verify the change in $f$ caused by the four endpoints $j_1, j_2, j_1'$, and $j_2'$. Suppose $y_1$ is returned, the contribution of these points in $f(y_1)$ is
	\begin{align*}
		& c_{j_1}y_1(F_{j_1}) + c_{j_2}y_1(F_{j_2})+ 
		c_{j_1'}y_1(F_{j_1'}) + c_{j_2'}y_1(F_{j_2'}) \\
		& =  c_{j_1}(y'(F_{j_1})+\delta_1) + c_{j_2}(y'(F_{j_2})-\delta_1) +
		c_{j_1'}\left(y(F_{j_1'})-\delta_1 \frac{\Delta_1}{\Delta_2}\right) + c_{j_2'}\left(y(F_{j_2'})+\delta_1\frac{\Delta_1}{\Delta_2}\right) \\
		&=  c_{j_1}y'(F_{j_1}) + c_{j_2}y'(F_{j_2})+ 
		c_{j_1'}y'(F_{j_1'}) + c_{j_2'}y'(F_{j_2'}) + \delta_1(c_{j_1}-c_{j_2})+ \delta_1\frac{\Delta_1}{\Delta_2}(c_{j_2'}-c_{j_1'}) \\
		&= c_{j_1}y'(F_{j_1}) + c_{j_2}y'(F_{j_2})+ 
		c_{j_1'}y'(F_{j_1'}) + c_{j_2'}y'(F_{j_2'}) .
	\end{align*}
Hence, $f(y_1) = f(y')$. Similarly, one can verify that $f(y_2) = f(y')$.
	
	\item Case $y'$ is rounded by \textsc{RoundFinalPath}: we have shown in the proof of Proposition \ref{prop:feasibility} that $y'(F_j) = 1$ for all $j \in J$ where $J$ is the set of vertices in $\R$ on the path $\vec{v}$. This fact and the other constraints of $\P_2$ ensure that $y'(F_j)$ is not decreasing for all $j \in V'$.

\end{itemize}
\end{proof}

\begin{proposition} $\Pr[j \in \C] \geq p_j$ for all $j \in V$.
\end{proposition}
\begin{proof} Let $y'$ be the vector defined as in lines $2$ and $3$ of \textsc{PseudoFRMCenterRound}. It suffices to show that, for all $j \in V'$, $\Pr[Y(F_j) = 1] \geq y'(F_j)$. (Note that $y'(F_j) \geq p_j$ by constraint (\ref{ineq:fair}).) This is because, for any vertex $k \in V \setminus V'$, the algorithm \textsc{RFiltering} guarantees that there exists $j \in V'$ such that $F_k \cap F_j \neq \emptyset$, and $y'(F_j) = \sum_{i \in B_j} x_{ij} \geq  \sum_{i \in V: d(i,k) \leq R} x_{ik} = y'(F_k)$. Notice that the event $Y(F_j)=1$ means there is some open center $F_j$ and the distance from $k$ to this center should be at most $3R$. Thus,  
$$\Pr[k \in \C] \geq \Pr[Y(F_j)=1] \geq y'(F_j) \geq y'(F_k) \geq p_k,$$ 
by constraint (\ref{ineq:fair}).

Fix any $j \in V'$. Recall that $Y$ is obtained by rounding $y'$ and, by Proposition \ref{prop:invariant2} and the proof of \ref{prop:feasibility}, we have $Y(F_j) \in \{0,1\}$ and $\Pr[Y(F_j)=1] = \E[Y(F_j)]$. We now show that the expected value of $y'(F_j)$ does not decrease after each iteration of the while-loop.
\begin{itemize}
	\item Case $y'$ is rounded by \textsc{RoundCycle}: $y'(F_j)$ is unchanged as before.
\item Case $y'$ is rounded by \textsc{RoundSinglePath}: if $j$ is not the endpoint of $\vec{v}$ then $y'(F_j)$ is unchanged. Otherwise, $y'(F_j)$ is increase by some $\delta_1 > 0$ with probability one.
	\item Case $y'$ is rounded by \textsc{RoundTwoPaths}: again, if $j \notin \{j_1, j_2, j_1', j_2'\}$ then $y'(F_j)$ is unchanged. Now suppose $j = j_1$. With probability $\delta_1/(\delta_1+\delta_2)$, $y'(F_{j_1})$ is increase by $\delta_2$, and, with the remaining probability, it is decreased by $\delta_1$. Thus, the expected change in $y'(F_{j_1})$ is 
	$$\frac{\delta_1}{\delta_1+\delta_2}(\delta_2) + \frac{\delta_2}{\delta_1+\delta_2}(-\delta_1) = 0. $$
Similarly, one can verify that the expected values of $y'(F_{j_2}), y'(F_{j_1'})$, and $ y'(F_{j_2'})$ remain the same.

	\item Case $y'$ is rounded by \textsc{RoundFinalPath}: we have showed in the proof of Proposition \ref{prop:feasibility} that if $j$ is on the path $\vec{v}$, then $Y(F_j) = 1$. Otherwise, the constraints of $\P_2$ ensure that $Y(F_j) = y'(F_j)$.
\end{itemize}

\end{proof}

So far we have proved the following theorem.
\begin{theorem}
\textsc{PseudoFRMCenterRound} will return a random solution $\S$ such that 
\begin{itemize}
	\item $\S$ is the union of some basis of $\M$ with (at most) one extra vertex,
	\item $|\C| \geq t$ with probability one,
	\item $\Pr[j \in \C] \geq p_j$ for all $j \in V$.
\end{itemize}
\label{thm:pseudo_matround}
\end{theorem}

\subsubsection{An algorithm satisfying the matroid constraint exactly}

Using a similar technique as in Section \ref{sec:knap_exact}, we will develop an approximation algorithm for the \textsf{FRMatCenter} problem which always returns a feasible solution. Let $\epsilon > 0$ a small parameter to be determined. Let $\U$ denote the collection of all possible sets of verticies with size at most $\lceil 1/\epsilon \rceil$ such that $U$ is an independent set of $\M$. Again, we have $|\U| \leq n^{O(1/\epsilon)}$. Suppose $R$ is the optimal radius to our instance. For any $i \in V$, recall that $\RBall(i, U, R)$ is the set of red vertices within radius $3R$ from $i$.

Consider the \emph{configuration} polytope $\P_\textsf{config3}$ containing points $(x,y,q)$ with the following constraints:
\begin{align*}
  \begin{cases}
      \sum_{U \in \U}q_U = 1     & \quad \\
      \sum_{i \in B_j} x^{U}_{ij} \leq q_U    & \quad \forall j \in V, U \in \U \\
      \sum_{U \in \U} \sum_{i \in B_j} x^{U}_{ij} \geq p_j    & \quad \forall j \in V \\
       x^{U}_{ij} \leq y^{U}_i   & \quad \forall i,j \in V, U \in \U\\
       \sum_{i \in W}  y^{U}_i \leq q_U r_\M(W)    & \quad  \forall U \in \U, W \subseteq V \\
       \sum_{j \in V} \sum_{i \in B_j} x^{U}_{ij} \geq q_U t   & \quad \\
       y^{U}_i = 1 & \quad \forall U \in \U, i \in U\\
       y^{U}_i = 0 & \quad \forall U \in \U, i \in V \setminus U, |\RBall(i, U, R)| \geq \epsilon n \\
       x^{U}_{ij}, y^{U}_i, q_U \geq 0   & \quad \forall i,j \in V, U \in \U \\
  \end{cases}
\end{align*}

We first claim that $\P_\textsf{config3}$ is a valid relaxation for the problem.

\begin{proposition} The polytope $\P_\textsf{config3}$ is non-empty.
\end{proposition}
\begin{proof}
Suppose $\S$ is a solution drawn from the optimal distribution $\D$. We compute a subset $U_\S$ of $\S$ using a similar procedure as in the proof of Proposition \ref{prop:valid_knap}. Recall that $ |\RBall(i, U_\S, R)| < \epsilon n$ for all $i \in \S \setminus U_\S$ and $|U_\S| \leq \lceil 1/\epsilon \rceil$. Since $U_\S \subseteq \S$, $U_\S$ is also an independent set of $\M$, implying that $U_\S \in \U$.

Now for any $U \in \U$, we set $q_U := \Pr[U_\S = U]$. Let $x^{U}_{ij}$ be probability of the joint event: $U_\S = U$ and $j$ is connected to $i$. Finally, let $y^{U}_i$ be the probability of the joint event: $U_\S = U$ and $i \in \S$. Then it is clear that $\sum_{U \in \U}q_U = 1$. Using similar arguments to the proofs of Propositions \ref{prop:valid_knap} and \ref{prop:relaxation_knapcenter}, we have the following inequalities:
\begin{align}
\sum_{i \in B_j} x^{U}_{ij} &\leq q_U,     \quad \forall j \in V, U \in \U \\
      \sum_{U \in \U} \sum_{i \in B_j} x^{U}_{ij} &\geq p_j,     \quad \forall j \in V \\
      \sum_{j \in V} \sum_{i \in B_j} x^{U}_{ij} &\geq q_U t,    \quad \\
        y^{U}_i &=  0  \quad \forall U \in \U, i \in V \setminus U, |\RBall(i, U, R)| \geq \epsilon n.
\end{align}
Recall that $y^U_i/q_U$ is the $\Pr[ i \in S \mid U = U_\S]$. Since $\S$ is independent with probability one, we have $  |\S \cap W | \leq r_\M(W)$ for all $W \subseteq V$. Therefore,
\begin{align*}
	r_\M(W) &\geq \E[|\S \cap W | ~|~ U = U_\S] = \sum_{i \in W}\Pr[i \in \S ~|~ U = U_\S] = \sum_{i \in W}y^U_i/q_U,
\end{align*}
for all $W \subseteq V$.

The other constraints can be verified easily. We conclude that $(x,y,q) \in \P_\textsf{config3}$.
\end{proof}

Next, let us pick any $(x,y,q) \in \P_\textsf{config3}$ and use the following algorithm to round it. 

\begin{algorithm}[h]
\caption{$\textsc{FRMCenterRound}\left(x,y,q \right)$}
\begin{algorithmic}[1]
\STATE Randomly pick a set $U \in \U$ with probability $q_U$
\STATE Let $x'_{ij} \gets x^{U}_{ij}/q_U$ and $y'_i \gets \min\{y_i^U/q_U, 1\}$
\STATE $\S' \gets \textsc{PseudoFRMCenterRound}\left(x', y' \right)$
\STATE Let $i^*$ be the ``extra'' vertex in $\S'$.
\RETURN $\S = \S' \setminus \{i\}$
\end{algorithmic} 
\label{algo:rfknapcenter}
\end{algorithm}

\bigskip \noindent \textbf{Analysis.} We are now ready to prove the second part of Theorem \ref{thm:FRMatCenter}. Let us fix any $\gamma > 0$ and set $\epsilon := \gamma^2$. Also, let $\e(U)$ be the event that Algorithm~\ref{algo:rfknapcenter} selects $U \in \U$. Note that $(x', y')$ satisfies the following inequalities:
\begin{align*}
	\sum_{j \in V} \sum_{i \in B_j} x_{ij}' &\geq t, \\
	\sum_{i \in B_j} x_{ij}' &\leq 1, \quad \forall j \in V, \\
	\sum_{i \in B_j} x_{ij}' &= \sum_{i \in B_j} x_{ij} / q_U, \quad  \forall j \in V, \\
	x_{ij}' &\leq y_i', \quad \forall i,j \in V, \\
	\sum_{i \in W}  y_i' &\leq  r_\M(W),    ~\quad \forall W \subseteq V.
\end{align*}

Moreover, every $i \in U$ satisfies $y'_i = 1$; every $i \in V \setminus U$ with $\RBall(i, U, R) \geq \epsilon n$ satisfies $y'_i = 0$.

Recall that the algorithm \textsc{PseudoFRMCenterRound} will return a solution $\S'$, which is the union of a basis of $\M$ with an extra center $i^*$. The vertex $i^*$ has $y_{i^*} \in (0,1)$, and in particular $i^* \notin U$. By removing $i^*$ from $\S'$, we ensure that the resulting set is a basis of $\M$ with probability one.

Now we shall prove the coverage guarantee. By Theorem \ref{thm:pseudo_matround}, $\S'$ covers at least $t$ vertices within radius $3R$. Any blue vertex can be connected to some center in $U$ and hence is not affected by the removal of $i^*$. Because $i^*$ covers at most $\epsilon n$ other red vertices, we have $$|\C| \geq t - \epsilon n = 1 - \gamma^2 n.$$

For any $j \in V$, let $X_j$ be the indicator random variable for the event that $d(j, \S') \leq 3R$ but $d(j, \S' \setminus \{i^* \} ) > 3 R$. We say that $j$ is a bad vertex iff $\E[X_j] \geq \gamma$, otherwise, $j$ is good. Again, $\sum_{j \in V} X_j \leq \epsilon n$ with probability one. Thus, there can be at most $\epsilon n / \gamma$ bad vertices. Letting $T$ be the set of good vertices, we have
$$|T| \geq n - \epsilon n/\gamma = (1-\gamma)n.$$ 
By Theorem \ref{thm:pseudo_matround}, $\Pr[j \text{ is covered by }\S'] \geq p_j$. So, for any $j \in T$, we have
\begin{align*}
	\Pr[j \in \C] \geq \Pr[j \text{ is covered by }\S'] - \Pr[X_j = 1] \geq p_j - \gamma.
\end{align*}

\clearpage
\bibliographystyle{plain}
\bibliography{Reference}

\end{document}